\documentclass[sigconf]{aamas}

\renewcommand\footnotetextcopyrightpermission[1]{} 
\usepackage{flushend}
\acmDOI{}  
\acmISBN{}  
\acmConference[ArXiV]{ArXiV}{November}{2020}  
\acmYear{2020}  

\usepackage{enumerate}
\usepackage{array}
\usepackage{amssymb}
\usepackage{amsmath}
\usepackage{mathrsfs}
\usepackage{amsthm}
\usepackage{newfloat}
\usepackage{mathtools}
\usepackage{amsmath}
\usepackage{empheq}
\usepackage{standalone}
\usepackage{tikz}
\usepackage{pgfplots}
\usepackage{subfig}
\usepackage{thmtools}  
\usepackage{thm-restate}
\usepackage{algorithm}
\usepackage[noend]{algpseudocode}
\usepackage{bbm}
\usepackage{multirow}
\newcommand{\X}{\mathcal{X}}
\newcommand{\Y}{\mathcal{Y}}
\newcommand{\bpi}{\boldsymbol{\pi}}
\DeclareMathOperator*{\argmax}{argmax}
\DeclareMathOperator*{\argmin}{argmin}
\algdef{SE}[DOWHILE]{Do}{DoWhile}{\algorithmicdo}[1]{\algorithmicwhile\ #1}%

\allowdisplaybreaks

\title[Learning PAC Maxmin Strategies in SBG with Infinite Strategy Spaces]{Learning Probably Approximately Correct Maximin Strategies in Simulation-Based Games with Infinite Strategy Spaces}

\author{Alberto Marchesi}
\affiliation{%
  \institution{Politecnico di Milano}
  \streetaddress{P.zza L. da Vinci}
  \city{Milano} 
  \state{Italia} 
  \postcode{20133}
}
\email{alberto.marchesi@polimi.it}
\author{Francesco Trov\`o}
\affiliation{%
	\institution{Politecnico di Milano}
	\streetaddress{P.zza L. da Vinci}
	\city{Milano} 
	\state{Italia} 
	\postcode{20133}
}
\email{francesco1.trovo@polimi.it}
\author{Nicola Gatti}
\affiliation{%
	\institution{Politecnico di Milano}
	\streetaddress{P.zza L. da Vinci}
	\city{Milano} 
	\state{Italia} 
	\postcode{20133}
}
\email{nicola.gatti@polimi.it}

\begin{document}

\begin{abstract}
	We tackle the problem of learning equilibria in \emph{simulation-based} games.
	In such games, the players' utility functions cannot be described analytically, as they are given through a black-box simulator that can be queried to obtain noisy estimates of the utilities.
	This is the case in many real-world games in which a complete description of the elements involved is not available upfront, such as complex military settings and online auctions.
	In these situations, one usually needs to run costly simulation processes to get an accurate estimate of the game outcome.
	As a result, solving these games begets the challenge of designing learning algorithms that can find (approximate) equilibria with high confidence, using as few simulator queries as possible.
	Moreover, since running the simulator during the game is unfeasible, the algorithms must first perform a pure exploration learning phase and, then, use the (approximate) equilibrium learned this way to play the game.   
	In this work, we focus on two-player zero-sum games with \emph{infinite strategy spaces}.
	Drawing from the best arm identification literature, we design two algorithms with theoretical guarantees to learn \emph{maximin} strategies in these games.
	The first one works in the \emph{fixed-confidence} setting, guaranteeing the desired confidence level while minimizing the number of queries.
	Instead, the second algorithm fits the \emph{fixed-budget} setting, maximizing the confidence without exceeding the given maximum number of queries.
	First, we formally prove $\delta$-PAC theoretical guarantees for our algorithms under some regularity assumptions, which are encoded by letting the utility functions be drawn from a Gaussian process.
	Then, we experimentally evaluate our techniques on a testbed made of randomly generated games and instances representing simple real-world security settings.
\end{abstract}

\keywords{Simulation-based games, Equilibrium computation}

\maketitle

\section{Introduction}\label{sec:intro}

Over the last two decades, game-theoretic models have received a growing interest from the AI community, as they allow to design artificial agents endowed with the ability of reasoning strategically in complex multi-agent settings.
This surge of interest was driven by many successful applications of game theory to challenging real-world problems, such as building robust protection strategies in security domains~\cite{tambe2011security}, designing truthful auctions for web advertising~\cite{gatti2015truthful}, and solving (\emph{i.e.}, finding the equilibria of) large zero-sum recreational games, \emph{e.g.}, Go~\cite{silver2016mastering}, different variants of Poker~\cite{brown2018superhuman,brown2019superhuman}, and Bridge~\cite{rong2019competitive}.  

Most of the game-theoretic studies in AI focus on models where a complete description of the game is available, \emph{i.e.}, the players' utilities can be expressed analytically.
This is the case of recreational games, which are commonly used as benchmarks for evaluating algorithms to compute equilibria in games~\cite{brown2017safe}.
However, in many real-world problems, the players' utilities may \emph{not} be readily available, as they are the outcome of a complex process governed by unknown parameters.
This is the case, \emph{e.g.}, in complex military settings where a comprehensive description of the environment and the units involved is not available, and online auctions in which the platform owner does not have complete knowledge of the parties involved. 
These scenarios can be addressed with \emph{simulation-based games} (SBGs)~\cite{vorobeychik2009strategic}, where the players' utilities are expressed by means of a black-box simulator that, given some players' strategies, can be queried to obtain a noisy estimate of the utilities obtained when playing such strategies.
These models beget new challenges in the design of algorithms to solve games: \emph{(i)} they have to learn (approximate) equilibria by using only noisy observations of the utilities, and \emph{(ii)} they should use as few queries as possible, since running the simulator is usually a costly operation.
Additionally, using the simulator while playing the game is unfeasible, since the simulation process might be prohibitively time consuming, as it is the case, \emph{e.g.}, in military settings where the units have to take prompt decisions when on the battlefield.
Thus, the algorithms must first perform a pure exploration learning phase and, then, use the (approximate) equilibrium learned this way to play the game. 

Despite the modeling power of SBGs, recent works studying such games are only sporadic, addressing specific settings such as, \emph{e.g.}, symmetric games with a large number of players~\cite{wiedenbeck2018regression,sokota2019learning}, empirical mechanism design~\cite{viqueira2019empirical}, and two-player zero-sum finite games~\cite{garivier2016maximin} (see Section~\ref{sec:related_works} for more details and additional related works).
To the best of our knowledge, the majority of these works focus on the case in which each player has a finite number of strategies available.
However, in most of the game settings in which simulations are involved, the players have an infinite number of choices available, \emph{e.g.}, physical quantities, such as angle of movement and velocity of units on a military field, bids in auctions, and trajectories in robot planning.
Dealing with infinite strategies leads to further challenges in the design of learning algorithms, since, being a complete exploration of the strategy space unfeasible, providing strong theoretical guarantees is, in general, a non-trivial task.
%


\subsection{Original Contributions}
We study the problem of learning equilibria in \emph{two-player zero-sum} SBGs with \emph{infinite strategy spaces}, providing theoretical guarantees.
Specifically, we focus on \emph{maximin} strategies for the first player, \emph{i.e.}, those maximizing her utility under the assumption that the second player acts so as to minimize it, after observing the first player's course of play.
For instance, this is the case in security games where a terrestrial counter-air defensive unit has to shoot an heat-seeking missile to a moving target that represents an approaching enemy airplane, which, after the attack has started, can respond to it by deploying an obfuscating flare with the intent of deflecting the missile trajectory.
When dealing with infinite strategy spaces, some regularity assumptions on the players' utilities are in order, since, otherwise, one cannot design learning algorithms with provable theoretical guarantees.
In this work, we encode our regularity assumptions on the utility function by modeling it as a sample from a \emph{Gaussian process} (GP)~\cite{williams2006gaussian}.
We design two algorithms able to learn (approximate) maximin strategies in two-player zero-sum SBGs with infinite strategy spaces, drawing from techniques used in the best arm identification literature.
The first algorithm we propose, called M-GP-LUCB, is for the \emph{fixed-confidence} setting, where the objective is to find an (approximate) maximin strategy with a given (high) confidence, using as few simulator queries as possible.
Instead, the second algorithm, called SE-GP, is for the \emph{fixed-budget} setting, in which a maximum number of queries is given in advance, and the task is to return an (approximate) maximin strategy with confidence as high as possible.  
First, we prove $\delta$-PAC (\emph{i.e.}, \emph{probably approximately correct}) theoretical guarantees for our algorithms in the easiest setting in which the strategy spaces are finite.
Then, we show how these results can be generalized to SBGs with infinite strategy spaces by leveraging the GP assumption.
Finally, we experimentally eventuate our algorithms on a testbed made of randomly generated games and instances based on the missile-airplane security game described above.
For SBGs with finite strategy spaces, we also compare our algorithms with the M-LUCB algorithm introduced by~\cite{garivier2016maximin} (the current state-of-the-art method for learning maximin strategies in two-player zero-sum finite games), showing that our methods dramatically outperform it.~\footnote{The complete proofs of our theoretical results are available in Appendices~A,~B,~and~C.}

\section{Related Works} \label{sec:related_works}

Over the last years, the problem of learning approximate equilibria in SBGs received considerable attention from the AI community.
In this section, we survey the main state-of-the-art works on the problem of learning equilibria in SBGs, highlighting which are the crucial differences with our work.
Let us remark that the majority of these works focus on SBGs with finite strategy spaces, while, to the best of our knowledge, ours provides the first learning algorithms with theoretical guarantees for SBGs with infinite strategy spaces.

The first computational studies on SBGs date back to the work of~\citet{vorobeychik2007learning}, who focus on $n$-player general-sum games, experimentally evaluating standard regression techniques to learn \emph{Nash equilibria} (NEs) in such games.
Their approach is to first learn the players' payoff functions and then compute an NE in the game learned this way.
\citet{gatti2011equilibrium} extend this work to sequential games.
Given the nature of regression techniques, this approach also works for SBGs with infinite strategy spaces.
However, the proposed methodology does not allow us to derive theoretical guarantees on the approximation quality of the obtained solutions, and it does not adopt any principled rule for choosing the next query to be performed.
%
%
In contrast, our algorithms are $\delta$-PAC, and, by exploiting techniques from the best arm identification literature, they also perform queries intelligently, allowing for a great reduction in the used number of queries.

A similar approach, which is still based on learning payoff functions using regression, is adopted by some recent works studying finite SBGs with many symmetric players~\cite{wiedenbeck2018regression,sokota2019learning}.
Their goal is to exploit the symmetries so as to learn symmetric NEs in large games efficiently.
\citet{wiedenbeck2018regression} focus on GP regression, since, as they show experimentally, it leads to better performances compared to other techniques.
Subsequently,~\citet{sokota2019learning} provide an advancement over the previous work, using neural networks to approximate the utility function (instead of GPs) and providing a way to guide sampling so as to focus it on the neighborhood of candidate equilibrium points.
These works significantly depart from ours, since \emph{(i)} they aim at finding symmetric NEs in large SBGs with many symmetric players, \emph{(ii)} they are restricted to games with finite strategy spaces, and (\emph{iii)} they do not provide any theoretical guarantees on the quality of the obtained solutions.

Recently, some works proposed learning algorithms for finite SBGs, relying on the PAC framework to prove theoretical guarantees~\cite{viqueira2019empirical,viqueira2019learning,wright2019probably}.
Specifically,~\citet{viqueira2019learning} and~\citet{wright2019probably} focus on learning NEs in $n$-player general-sum games.
However, their results are limited to the case of finite strategy spaces and cannot be easily generalized to settings involving infinite strategy spaces, as they do not introduce any regularity assumption on the players' utility functions.
Moreover, the querying algorithms they propose are based on a global exploration of the strategy profiles of the game, without relying on specific selection rules, except for the elimination of sub-optimal strategies. 
In contrast, our algorithms exploit best arm identification techniques, and, thus, they employ principled selection rules that allow to focus queries on the most promising strategy profiles. 
%

It is also worth pointing out some works that, while being not directly concerned with SBGs, address related problems.
Recently, a growing attention has been devoted to no-regret learning algorithms in games with bandit feedback~\cite{kleinberg2005nearly,flaxman2005online,bravo2018bandit}.
The methods developed in this framework are significantly different from ours, as they fit the classical multi-armed bandit scenario where the objective is to minimize the cumulative regret.
Instead, we adopt the best arm identification perspective, where the goal is to identify an optimal arm with high confidence.
%
Thus, our querying algorithms might achieve large regret during the learning process, since they are focused on a pure exploration task in which exploitation is not a concern. 
Moreover, no-regret learning algorithms require strong assumptions to converge to equilibria in games with bandit feedback (such as, \emph{e.g.}, concavity of the players' utility functions~\cite{bravo2018bandit}).
In contrast, our theoretical guarantees do not need any explicit requirement on the utilities (except for a reasonable degree of smoothness, encoded by the GP assumption), and, thus, they also hold when the players' utility functions exhibit a complex (\emph{e.g.}, non-concave) dependence on the players' strategies.

There are also other related problems not directly connected with SBGs that are worth citing, such as, \emph{e.g.}, meta-game analysis~\cite{tuyls2018generalised}, learning unknown game parameters or players' rationality models by observing played actions~\cite{ling2018game,ling2019large}, combining supervised learning techniques with decision-making in uncertain optimization models~\cite{wilder2019melding}, and online learning in games~\cite{bisi2017regret}.

\section{Preliminaries}\label{sec:preliminaries}

A \emph{two-player zero-sum game with infinite strategy spaces} is a tuple $\Gamma = (\X, \Y, u)$, where
$\X \subset \mathbb{R}^d$ and $\Y \subset \mathbb{R}^d$ are compact and convex sets of strategies available to the first and the second player, respectively, while $u : \X \times \Y \mapsto \mathbb{R}$ is a function defining the utility for the first player.~\footnote{For the ease of presentation, in the following we focus on the case in which $\X \subset [0,1]$ and $\Y \subset [0,1]$ are closed intervals. The generalization of our results to the case in which the strategy spaces are compact and convex subsets of $\mathbb{R}^d$ is straightforward.}
Since the game is zero-sum, the second player's utility is given by $-u$.
A \emph{two-player zero-sum game with finite strategy spaces} is defined analogously, with $\X$ and $\Y$ being finite sets, \emph{i.e.}, $\X \coloneqq \{ x^1, \ldots, x^n \}$ and $\Y \coloneqq \{ y^1, \ldots, y^m \}$, with $n > 1$ and $m > 1$ denoting the finite numbers of strategies available to the first and the second player, respectively.
For the ease of notation, letting $\Pi \coloneqq \X \times \Y$, we denote with $\boldsymbol{\pi} \coloneqq (x,y) \in \Pi$ a \emph{strategy profile}, \emph{i.e.}, a tuple specifying a strategy $x \in \X$ for the first player and a strategy $y \in \Y$ for the second player.

In this work, we are concerned with the computation of \emph{maximin} strategies, adopting the perspective of the first player.
In words, we seek for a first player's strategy that maximizes her utility, assuming a worst-case opponent that acts so as to minimize it.
Since the game is zero sum, we can assume that the second player decides how to play after observing the first player's move, and, thus, playing a maximin strategy is the best choice for the first player.~\footnote{This assumption is in line with the classical Stackelberg model in which the second player (follower) gets to play after observing the strategy of the first one (leader)~\cite{stackelberg1934marktform}.}
Formally, given a first player's strategy $x \in \X$, we denote with $y^\ast(x) \in \arg \min_{y \in \Y} u(x,y)$ a second player's best response to $x$.
%
Then, $x^\ast \in \X$ is a maximin strategy for the first player if $x^\ast \in \arg \max_{x \in \X} u(x,y^\ast(x))$, with $\bpi^\ast \coloneqq (x^\ast, y^\ast(x^\ast))$ denoting its corresponding maximin strategy profile.~\footnote{Even though playing a maximin strategy may not be the optimal choice for the first player if the players are assumed to play simultaneously, this is the case if we require additional (mild) technical assumptions guaranteeing that $\bpi^\ast$ is an equilibrium point of the game; see~\cite{sion1958general} for additional details.}

\subsection{Simulation-Based Games}
In SBGs, the utility function $u$ is not readily available, but it is rather specified by an exogenous simulator that provides noisy point estimates of it.
As a result, in SBGs, one cannot explicitly compute a maximin strategy, and, thus, the problem is to learn one by sequentially querying the simulator.
At each round $t$, the simulator is given a strategy profile $\boldsymbol{\pi}_t \in \Pi$ and returns an estimated utility $\tilde{u}_t \coloneqq u(\boldsymbol{\pi}_t) + e_t$, where $e_t \sim \mathcal{N}(0, \lambda)$ is i.i.d.~Gaussian noise.
The goal is to find a good approximation (see Equation~\eqref{eq:delta_pac}) of a maximin strategy $x^\ast \in \X$ as rapidly as possible, \emph{i.e.}, limiting the number of queries to the simulator.
To achieve this, we follow the approach of~\citet{garivier2016maximin} and propose some \emph{dynamic querying algorithms} (see Algorithm~\ref{alg:interaction}, where $\textsc{Sim}(\bpi)$ represents a simulator query for $\bpi \in \Pi$), which are generally characterized by:
\begin{itemize}
	\item a querying rule that indicates which strategy profile $\boldsymbol{\pi}_t \in \Pi$ is sent as input to the simulator at each round $t$;
	\item a stopping rule that determines the round $T$ after which the algorithm terminates its execution;
	\item a final guess $\overline{\bpi} \coloneqq (\bar{x},\bar{y}) \in \Pi$ for the (true) maximin strategy profile $\bpi^\ast$ of the game.
\end{itemize}

\begin{algorithm}[t!]
	\centering
	\caption{Dynamic Querying Algorithm}
	\begin{algorithmic}[1]
		\State $t \gets 1$
		\Do
			\State \parbox[t]{\dimexpr\textwidth-\leftmargin-\labelsep-\labelwidth}{%
				Select a strategy profile $\bpi_{t} \in \Pi$ according \newline  \phantom{aaa} to the querying rule \strut}
			\State Get estimated utility $\tilde{u}_{t} \gets \textsc{Sim}(\bpi_{t})$
			\State Update the algorithm parameters using $\tilde{u}_{t}$
			\State $t \gets t + 1$
		\DoWhile{stopping condition is not met}
		\State\Return final guess $\overline{\bpi} = (\bar{x}, \bar{y})$ for the maximin profile
	\end{algorithmic}
	\label{alg:interaction}
\end{algorithm}

Given a desired approximation $\epsilon \geq 0$, the objective of the algorithm is to find an $\epsilon$-maximin strategy with high accuracy, using as few queries as possible to the simulator.
Formally, given $\delta \in (0,1)$, our goal is to design algorithms that are $\delta$-PAC, \emph{i.e.}, they satisfy:
\begin{equation} \label{eq:delta_pac}
	\forall u \quad \mathbb{P} \Big( \left| u(\bpi^\ast) - u(\bar{x}, y^\ast(\bar{x})) \right| \leq \epsilon \Big) \geq 1 - \delta,
\end{equation}
while keeping the number of rounds $T$ as small as possible.
This is known as the \emph{fixed-confidence} setting (see Section~\ref{sec:fixed_confidence}).
An alternative is to consider the \emph{fixed-budget} case, where the maximum number of rounds $T$ is given in advance, and the goal is to minimize the probability $\delta$ that $\bar{x}$ is not an $\epsilon$-maximin strategy (see Section~\ref{sec:fixed_budget}).
Notice that, for SBGs with finite strategy spaces, the $\delta$-PAC property in Equation~\eqref{eq:delta_pac} can only require $u(\bpi^\ast) - u(\bar{x}, y^\ast(\bar{x})) \leq \epsilon$, since it is always the case that $u(\bpi^\ast) > u(\bar{x}, y^\ast(\bar{x}))$.~\footnote{This is in line with the definition provided by~\citet{garivier2016maximin}.}

%

\subsection{Gaussian Processes}
To design $\delta$-PAC algorithms working with SBGs having infinite strategy spaces, we first need to introduce some regularity assumptions on the utility functions $u$.
In this work, we model the utility as a sample from a GP, which is a collection of dependent random variables, one for each action profile $\boldsymbol{\pi} \in \Pi$, every finite subset of which is multivariate Gaussian distributed~\cite{williams2006gaussian}.
A GP$(\mu(\bpi),k(\bpi,\bpi'))$ is fully specified by its \emph{mean} function $\mu : \Pi \mapsto \mathbb{R}$, with $\mu(\boldsymbol{\pi}) \coloneqq \mathbb{E}[u(\boldsymbol{\pi})]$, and its \emph{covariance} (or \emph{kernel}) function $k : \Pi \times \Pi \mapsto \mathbb{R}$, with $k(\boldsymbol{\pi},\boldsymbol{\pi}') \coloneqq \mathbb{E}[(u(\boldsymbol{\pi}) - \mu(\boldsymbol{\pi}))(u(\boldsymbol{\pi}') - \mu(\boldsymbol{\pi}'))]$.
W.l.o.g., we assume that $\mu \equiv \mathbf{0}$ and the variance is bounded, \emph{i.e.}, $k(\bpi, \bpi) \coloneqq \sigma^2 \leq 1$ for every $\bpi \in \Pi$.
Note that the GP assumption guarantees that the utility function $u$ has a certain degree of smoothness, without relying on rigid parametric assumptions, such as linearity.
Intuitively, the kernel function $k$ determines the correlation of the utility values across the space of strategy profiles $\Pi$, thus encoding the smoothness properties of the utility functions $u$ sampled from GP$(\mu(\bpi), k(\bpi,\bpi'))$ (for some examples of commonly used kernels, see Section~\ref{sec:experiments}).

We also need GPs in our algorithms, as they use GP$(\mathbf{0}, k(\bpi,\bpi'))$ as prior distribution over $u$.
The major advantage of working with GPs is that they admit simple analytical formulas for the mean and covariance of the posterior distribution.
These relations can be easily expressed using matrix notation, as follows.
Let $\tilde{\mathbf{u}}_t \coloneqq [\tilde{u}_1, \ldots, \tilde{u}_t]^\top$ be the vector of utility values observed up to round $t$, obtained by querying the simulator on the strategy profiles $\bpi_1, \ldots, \bpi_t$.
Then, the posterior distribution over $u$ is still a GP, with mean $\mu_t(\bpi)$, covariance $k_t(\bpi,\bpi')$, and variance $\sigma_t^2(\bpi)$, which are defined as follows:
\begin{align}
	\mu_t(\bpi) & \coloneqq \mathbf{k}_t(\bpi)^\top \left( K_t + \lambda I \right)^{-1} \tilde{\mathbf{u}}_t, \label{eq:gp_posterior_1}\\
	k_t(\bpi, \bpi') & \coloneqq k(\bpi,\bpi') - \mathbf{k}_t(\bpi)^\top \hspace{-1mm}\left( K_t + \lambda I \right)^{-1}  \mathbf{k}_t(\bpi'),\label{eq:gp_posterior_2} \\
	\sigma_t^2(\bpi) & \coloneqq k_t(\bpi,\bpi),\label{eq:gp_posterior_3}
\end{align}
where $\mathbf{k}_t(\bpi) \coloneqq [ k(\bpi,\bpi_1), \ldots, k(\bpi, \bpi_t) ]^\top$ and $K_t$ is the positive definite $t \times t$ kernel matrix, whose $(i,j)$-th entry is $k(\bpi_i, \bpi_j)$.
The posterior parameters update formulas can also be expressed recursively, thus avoiding costly matrix inversions, as shown in~\cite{chowdhury2017kernelized}.
Letting $\bpi_t$ and $\tilde{u}_t$ be, respectively, the queried strategy profile and the observed utility at round $t$, we can write:
\begin{align}
	&\mu_t(\bpi) \gets \mu_{t-1}(\bpi) + \frac{k_{t-1}(\bpi, \bpi_t)}{\lambda + \sigma^2_{t-1}(\bpi_t)} (\tilde{u}_t - \mu_{t-1}(\bpi_t)), \label{eq:gp_posterior_1_rec}\\
	&k_t(\bpi, \bpi') \gets k_t(\bpi, \bpi') - \frac{k_{t-1}(\bpi, \bpi_t) k_{t-1}(\bpi_t, \bpi')}{\lambda + \sigma^2_{t-1}(\bpi_t)}, \label{eq:gp_posterior_2_rec}\\
	&\sigma_t^2(\bpi) \gets \sigma_{t-1}^2(\bpi) - \frac{k_{t-1}^2(\bpi, \bpi_t)}{\lambda + \sigma^2_{t-1}(\bpi_t)}. \label{eq:gp_posterior_3_rec}
\end{align}
Clearly, at the beginning of the algorithms, the estimates are initialized using the GP prior GP$(\mathbf{0},k(\bpi,\bpi'))$, \emph{i.e.}, formally, $\mu_0(\bpi) \coloneqq \mathbf{0}$, $k_0(\bpi, \bpi') \coloneqq k(\bpi, \bpi')$, and $\sigma_0^2(\bpi) \coloneqq k(\bpi, \bpi) = \sigma^2$.

\section{Fixed-Confidence Setting}\label{sec:fixed_confidence}

In this section and the following one (Section~\ref{sec:fixed_budget}), we present our learning algorithms for the easiest setting of SBGs with finite strategy spaces.
Then, in Section~\ref{sec:continuous}, we show how they can be extended to SBGs with infinite strategy spaces. 

For the fixed-confidence setting, we propose a $\delta$-PAC dynamic querying algorithm (called M-GP-LUCB, see Algorithm~\ref{alg:m_gplucb}) based on the M-LUCB approach introduced by~\citet{garivier2016maximin} and provide a bound on the number of rounds $T_\delta$ it requires, as a function of the confidence level $\delta$.
While our algorithm shares the same structure as M-LUCB, it uses confidence bounds relying on the GP assumption, and, thus, different proofs are needed to show its $\delta$-PAC properties.
As shown in Section~\ref{sec:continuous}, our algorithm and its theoretical guarantees have the crucial advantage of being easily generalizable to SBGs with infinite strategy spaces.

For every strategy profile $\bpi \in \Pi$, the algorithm keeps track of a confidence interval $[L_t(\bpi), U_t(\bpi)]$ on $u(\bpi)$ built using the utility values $\tilde{u}_t$ observed from the simulator up to round $t$.
Using GP$(\mathbf{0}, k(\bpi,\bpi'))$ as prior distribution over the utility function $u$, the lower bounds of the intervals are defined as $L_t(\bpi) \coloneqq \mu_t(\bpi) - \sqrt{b_t} \sigma_t(\bpi)$ and the upper bounds as $U_t(\bpi) \coloneqq \mu_t(\bpi) + \sqrt{b_t} \sigma_t(\bpi)$, where $\mu_t$ and $\sigma_t^2$ are the mean and the variance of the posterior distribution computed with observations up to round $t$ (see Equations~\eqref{eq:gp_posterior_1}--\eqref{eq:gp_posterior_3}), while $b_t$ is an exploration term that depends from the context (see Theorem~\ref{thm:fixed_confidence}).

At the end of every even round $t$, the algorithm selects the strategy profiles to give as inputs to the simulator during the next two rounds $t+1$ and $t+2$. 
For every $x \in \X$, let 
\begin{equation*}
	\gamma_t(x) \coloneqq \argmin_{y \in \Y} L_t(x,y)
\end{equation*}
be the second player's best response to $x$ computed using the lower bounds $L_t$.
Moreover, let
\begin{equation*}
	\bar{x}_t \coloneqq \argmax_{x \in \X} \min_{y \in \Y} \mu_t(x,y)
\end{equation*}
be the maximin strategy computed using the posterior mean $\mu_t$.
Then, in the following two rounds, the algorithm selects the strategy profiles $\bpi_{t+1}$ and $\bpi_{t+2}$, defined as follows:
\begin{align}
	&\bpi_{t+1} \coloneqq (\bar{x}_t, \gamma_t(\bar{x}_t)) \label{eq:select_m_lucb_1} \\
	&\bpi_{t+2} \coloneqq \argmax_{\bpi \in \{ (x, \gamma_t(x)) \}_{x \neq \bar{x}_t}}  U_t(\bpi). \label{eq:select_m_lucb_2}
\end{align}
This choice is made so as to advance the algorithm towards its termination.
In particular, the M-GP-LUCB algorithm stops when, according to the confidence intervals, the strategy $\bar{x}_t$ is probably approximately better than all the others, \emph{i.e.}, when it holds $L_t(\bpi_{t+1}) > U_t(\bpi_{t+2}) - \epsilon$.
%
%
Intuitively, $\bpi_{t+1}$ represents the best candidate for being a maximin strategy profile, while $\bpi_{t+2}$ is the second-best candidate.
Thus, the algorithm stops if $\bpi_{t+1}$ is better than $\bpi_{t+2}$ with sufficiently high confidence, \emph{i.e.}, whenever the lower bound for the former is larger than the upper bound for the latter (up to an approximation of $\epsilon$).
The final strategy profile recommended by the algorithm is $\overline{\bpi} \coloneqq (\bar{x}_t, \gamma_t(\bar{x}_t))$.

\begin{algorithm}[t!]
	\centering
	\caption{\textsc{M-GP-LUCB}($\epsilon$, $\delta$)}
	\begin{algorithmic}[1]
		\State Initialize $t \gets 0$, $\mu_0(\bpi) \gets \mathbf{0}$, $k_0(\bpi,\bpi') \gets k(\bpi,\bpi')$
		\Do
			\State Select $\bpi_{t+1}$ and $\bpi_{t+2}$ using Eqs.~\eqref{eq:select_m_lucb_1}--\eqref{eq:select_m_lucb_2}
			\State $\tilde{u}_{t+1} \gets \textsc{Sim}(\bpi_{t+1})$, $\tilde{u}_{t+2} \gets \textsc{Sim}(\bpi_{t+2})$
			\State \parbox[t]{\dimexpr\textwidth-\leftmargin-\labelsep-\labelwidth}{%
				Compute $\mu_{t+2}(\bpi)$ and $k_{t+2}(\bpi,\bpi')$ using \newline \phantom{aaa} observations $\tilde{u}_{t+1}$, $\tilde{u}_{t+2}$ and Eqs.~\eqref{eq:gp_posterior_1_rec}--\eqref{eq:gp_posterior_3_rec} \strut}
			\State $t \gets t + 2$
		\DoWhile{$L_t(\bpi_{t+1}) \leq U_t(\bpi_{t+2}) - \epsilon$}
		\State\Return $\overline{\bpi} = (\bar{x}_t, \gamma_t(\bar{x}_t))$
	\end{algorithmic}
	\label{alg:m_gplucb}
\end{algorithm}

The following theorem shows that M-GP-LUCB is $\delta$-PAC and provides an upper bound on the number of rounds $T_\delta$ it requires.
The analysis is performed for $\epsilon = 0$, \emph{i.e.}, when $\overline{\bpi}$ is evaluated with respect to an exact maximin profile.~\footnote{Assuming $\epsilon = 0$ also requires the additional w.l.o.g.~assumption that the utility value of an exact maximin strategy and that one of a second-best maximin strategy are different.}
Note that the upper bound for $T_\delta$ depends on the utility-dependent term $H^\ast(u) \coloneqq \sum_{\bpi \in \Pi}c (\bpi)$, where, for $\bpi = (x,y) \in \Pi$, $c(\bpi)$ is defined as follows:
\begin{equation*}
	c(\bpi) \coloneqq \frac{1}{\max \left\{(\Delta^*)^2, \left(\frac{u(x^\ast, y^\ast(x^\ast)) + u(x^{\ast\ast}, y^\ast(x^{\ast\ast}))}{2} - u(x,y^\ast(x)) \right)^2 \right\} },
\end{equation*}
where, for the ease of writing, we let $\Delta^* \coloneqq u(\bpi) - u(x, y^\ast(x))$ and $x^{\ast\ast} \in \argmax_{x \in \X \setminus \{x^\ast\}} u(x, y^\ast(x))$, \emph{i.e.}, $x^{\ast\ast}$ is a first player's maximin strategy when $x^\ast$ is removed from the available ones.
This term has the same role as $H_1 \coloneqq \sum_{i \in \{1, \ldots, |\Pi|\}} \frac{1}{\Delta_{(i)}^2}$ used by~\citet{audibert2010best} in the best arm identification setting, where $\bpi^i$ is the $i$-th strategy profile in $\Pi$, which is ordered in such a way that, letting $\Delta_{(i)} \coloneqq |u(\bpi^*) - u(\bpi^i)|$, it holds $\Delta_{(1)} \leq \Delta_{(2)} \leq \ldots \leq \Delta_{(|\Pi|)}$. 
%
%
Intuitively, $H^*(u)$ and $H_1$ characterize the hardness of the problem instances by determining the amount of rounds required to identify the maximin profile and the best arm, respectively.

\begin{restatable}[]{theorem}{thmconfidence}\label{thm:fixed_confidence}
	Using a generic nondecreasing exploration term $b_t > 0$, the M-GP-LUCB algorithm stops its execution after at most $T_\delta$ rounds, where:
	\begin{equation}
		T_\delta \leq \inf \left\{t \in \mathbb{N} \ : \ 8 \,H^*(u)\, b_t \,\lambda - \frac{\lambda \,n\, m}{\sigma^2} < t \right\}. \label{eq:fixed_conf}
	\end{equation}
	Specifically, letting $b_t := 2 \log \left( \frac{n\, m\, \pi^2\, t^2}{6 \delta} \right)$, the algorithm returns a maximin profile with confidence at least $1 - \delta$, and:
	\begin{align}
		T_\delta \leq 64 \,H^*(u)\, \lambda \left( \log \left( 64 \,H^*(u)\, \lambda \,\pi \sqrt{\frac{n\, m}{6 \delta}} \right) +  \right. \nonumber\\
	\left. + 2 \log \left( \log \left( 64 \,H^*(u) \lambda \pi \sqrt{\frac{n\, m}{6 \delta}} \right) \right) \right), \label{eq:fixed_conf_final}
	\end{align}
	where we require that $64 \,\lambda \,\pi \,\sqrt{\frac{n\, m}{6 \delta}} > 4.85$.
\end{restatable}

Intuitively, from the result in Theorem~\ref{thm:fixed_confidence}, we can infer that the most influential terms on the number of rounds required to get a specific confidence level $\delta$ are $H^*(u)$ and the noise variance $\lambda$, which impact as multiplicative constants on $T_\delta$.
On the other hand, $T_\delta$ scales only logarithmically with the number of strategy profiles $|\Pi| = n\, m$, thus allowing the execution of the M-GP-LUCB algorithm also in settings where the players have a large number of strategies available.

\begin{algorithm}[t!]
	\centering
	\caption{\textsc{GP-SE}($T$)}
	\begin{algorithmic}[1]
		\State Initialize $\Pi_1 \gets \Pi$, $\mu_0(\bpi) \gets \mathbf{0}$
		\For{$p = 1,2, \ldots, P-1$}
		\State \parbox[t]{\dimexpr\textwidth-\leftmargin-\labelsep-\labelwidth}{%
			For each $\bpi \in \Pi_p$, query $\textsc{Sim}(\bpi)$ for \newline \phantom{aaa} $T_p - T_{p-1}$ rounds \strut}
		\State Compute $\mu_p(\bpi)$ using observations
		\State Select $\bpi_p$ according to Eqs.~\eqref{eq:xfixbud}--\eqref{eq:yfixbud}
		\State $\Pi_{p+1} \gets \Pi_p \setminus \{ \bpi_p \}$
		\EndFor
		\State\Return the unique element $\overline{\bpi}$ of $\Pi_{P}$
	\end{algorithmic}
	\label{alg:gp_se}
\end{algorithm}

\section{Fixed-Budget Setting}\label{sec:fixed_budget}

In the fixed-budget setting, the goal is to design $\delta$-PAC algorithms that, given the maximum number of available rounds $T$ (\emph{i.e.}, the budget), find an $\epsilon$-maximin strategy with confidence $1 - \delta_T$ as large as possible.
We propose a successive elimination algorithm (called GP-SE, see Algorithm~\ref{alg:gp_se}), which is based on an analogous method proposed by~\citet{audibert2010best} for the best arm identification problem.
The fundamental idea behind our GP-SE algorithm is a novel elimination rule, which is suitably defined for the problem of identifying maximin strategies.

The algorithm works by splitting the number of available rounds $T$ into $P - 1$ phases, where, for the ease of notation, we let $P \coloneqq |\Pi| = n\,m$ be the number of players' strategy profiles.
At the end of each phase, the algorithm excludes from the set of candidate solutions the strategy profile that has the lowest chance of being maximin.
Specifically, letting $\Pi_p$ be the set of the remaining strategy profiles during phase $p$, at the end of $p$, the algorithm dismisses the strategy profile $\bpi_p \coloneqq (x_p, y_p) \in \Pi_p$, defined as follows:
\begin{align}
	(x_p, \boldsymbol{\cdot}) & \coloneqq \argmin_{\bpi \in \Pi_p} \mu_p (\bpi), \label{eq:xfixbud}\\
	y_p & \coloneqq \argmax_{y \in \Y : (x_p, y) \in \Pi_p} \mu_p (x_p, y), \label{eq:yfixbud}
\end{align}
where $\mu_p$ represents the mean of the posterior distribution computed at the end of phase $p$ (see Equations~\eqref{eq:gp_posterior_1}-\eqref{eq:gp_posterior_3}).
Intuitively, the algorithm selects the first player's strategy $x_p$ that is less likely to be a maximin one, together with the second player's strategy $y_p$ that is the worst for her given $x_p$.
%
%
At the end of the last phase, the (unique) remaining strategy profile $\overline{\bpi} = (\bar{x}, \bar{y})$ is recommended by the algorithm.

Following~\cite{audibert2010best}, the length of the phases have been carefully chosen so as to obtain an optimal (up to a logarithmic factor) convergence rate.
Specifically, letting $\overline{\log}(P) \coloneqq \frac{1}{2} + \sum_{i = 2}^{P} \frac{1}{i}$, let us define $T_0 \coloneqq 0$ and, for every phase $p \in \{1, \ldots, P-1 \}$, let:
\begin{equation}
	T_p \coloneqq \left\lceil  \frac{T - P}{\overline{\log}(P) (P + 1 -p )} \right\rceil.
\end{equation}
Then, during each phase $p$, the algorithm selects every remaining strategy profile in $\Pi_p$ for exactly $T_p - T_{p-1}$ rounds.
Let us remark that the algorithm is guaranteed to do not exceed the number of available rounds $T$.
Indeed, each $\bpi_p$ is selected for $T_p$ rounds, while $\overline{\bpi}$ is chosen $T_{P - 1}$ times, and $\sum_{p = 1}^{P-1} T_{p} + T_{P-1} \leq T$ holds by definition.

The following theorem provides an upper bound on the probability $\delta_T$ that the strategy profile $\overline{\bpi}$ recommended by the GP-SE algorithm is \emph{not} $\epsilon$-maximin, as a function of the number of rounds $T$.
As for the fixed-confidence setting, our result holds for $\epsilon = 0$.

\begin{restatable}[]{theorem}{thmse} \label{thm:fixed_budget}
	Letting $T$ be the number of available rounds, the GP-SE algorithm returns a maximin strategy profile $\bpi^\ast$ with confidence at least $1 - \delta_T$, where:
	\begin{equation} \label{eq:fixed_bud}
		\delta_T \coloneqq 2 P (n + m - 2) e^{-\frac{T - P}{8 \lambda \overline{\log}(P) H_2}},
	\end{equation}
	and $H_2 \coloneqq \max_{i \in \{1, \ldots, P\}} i \, \Delta_{(i)}^{-2}$.
\end{restatable}
As also argued by~\citet{audibert2010best}, a successive elimination method provides two main advantages over a simple round robin querying strategy in which every strategy profile is queried for the same number of rounds. 
First, it provides a similar bound on $\delta_T$ with a better dependency on the parameters, and, second, it queries the maximin strategy profile a larger number of times, thus returning a better estimate of its expected utility.

\section{Simulation-Based Games with Infinite Strategy Spaces}\label{sec:continuous}

We are now ready to provide our main results on SBGs with infinite strategy spaces.
In the first part of the section, we show how the $\delta$-PAC algorithms proposed in Sections~\ref{sec:fixed_confidence}~and~\ref{sec:fixed_budget} for finite SBGs can be adapted to work with infinite strategy spaces while retaining some theoretical guarantees on the returned $\epsilon$-maximin profiles.
This requires to work with a (finite) discretized version of the original (infinite) SBGs, where the players' strategy spaces are approximated with grids made of equally spaced points.
Then, in the second part of the section, we provide some results for the situations in which one cannot work with this kind of discretization, and, instead, only a limited number of points is sampled from the players' strategy spaces.
This might be the case when, \emph{e.g.}, the dimensionality $d$ of the players' strategy spaces is too high, or there are some constraints on the strategy profiles that can be queried.
Clearly, in this setting, we cannot prove $\delta$-PAC results, as the quality of the $\epsilon$-maximin strategy profiles inevitable depends on how the points are selected.

Let us remark that our main results rely on our assumption that the utility function $u$ is drawn from a GP, provided some mild technical requirements are satisfied (see Assumption~\ref{ass:smooth}).

\subsection{$\delta$-PAC Results for Evenly-Spaced Grids}

The idea is to work with a discretization of the players' strategy spaces, each made of at least $K_\epsilon$ equally spaced points, where $\epsilon \geq 0$ is the desired approximation level.  
This induces a new (restricted) SBGs with finite strategy spaces, where techniques presented in the previous sections can be applied.
In the following, for the ease of presentation, given an SBG with infinite strategy spaces $\Gamma$, we denote with $\Gamma(K)$ the finite SBG obtained when approximating the players' strategy spaces with $K$ equally spaced points, \emph{i.e.}, a game in which the players have $n = m = K$ strategies available and the utility value of each of the $n\,m$ strategy profiles is the same as that one of the corresponding strategy profile in $\Gamma$.
First, let us introduce the main technical requirement that we need for our results to hold.

\begin{restatable}[Kernel Smoothness]{assumption}{asssmooth}\label{ass:smooth}
	A kernel $k(\bpi,\bpi')$ is said to be \emph{smooth} over $\Pi$ if, for each $L > 0$ and for some constants $a, b > 0$, the functions $u$ drawn from \emph{GP}$(\mathbf{0},k(\bpi,\bpi'))$ satisfy:
	\begin{align}
		&\mathbb{P} \left( \sup_{\bpi \in \Pi} \left| \frac{\partial u}{\partial x} \right| > L \right) \leq a e^{- \frac{L^2}{b^2}}\\
		&\mathbb{P} \left( \sup_{\bpi \in \Pi} \left| \frac{\partial u}{\partial y} \right| > L \right) \leq a e^{- \frac{L^2}{b^2}}.
	\end{align}
\end{restatable}
This assumption is standard when using GPs in online optimization settings~\cite{srinivas2009gaussian}, and it is satisfied by many common kernel functions for specific values of $a$ and $b$, such as the squared exponential kernel and the Mat\'ern one with smoothness parameter $\nu > 2$ (see Section~\ref{sec:experiments} for details on the definition of these kernels).

We are now ready to state our main result:
\begin{restatable}{theorem}{thmcont}\label{thm:cont}
	Assume that $u$ is drawn from a \emph{GP}$(\mathbf{0},k(\bpi,\bpi'))$ satisfying Assumption~\ref{ass:smooth}.
	Given $\epsilon > 0$ and $\delta \in (0, 2)$, let $\overline{\bpi} \coloneqq (\bar{x}, \bar{y})\in \Pi$ be a maximin strategy profile for a finite game $\Gamma(K)$ where $K$ is at least $K_\epsilon \coloneqq \left\lceil \frac{b}{2\epsilon} \sqrt{\log \left( \frac{4a}{\delta}\right)} \, \right\rceil + 1$.
	Then, the following holds:
	\begin{equation}
		\mathbb{P} \Big( \left| u(\bpi^\ast) - u(\overline{\bpi}) \right| \leq \epsilon \Big) \geq 1 - \frac{\delta}{2}.
	\end{equation}
\end{restatable}

The following two results rely on Theorem~\ref{thm:cont} to show that the M-GP-LUCB (Algorithm~\ref{alg:m_gplucb}) and the GP-SE (Algorithm~\ref{alg:gp_se}) algorithms can be employed to find, with high confidence, $\epsilon$-maximin strategy profiles in SBGs with infinite strategy spaces.
Let us remark that, while for SBGs with finite strategy spaces our theoretical analysis is performed for $\epsilon = 0$, in the case of infinite strategy spaces it is necessary to assume a nonzero approximation level $\epsilon$.

\begin{restatable}{corollary}{corconfidence}\label{cor:cont_conf}
	Assume that $u$ is drawn from a \emph{GP}$(\mathbf{0},k(\bpi,\bpi'))$ satisfying Assumption~\ref{ass:smooth}.
	Given $\epsilon > 0$ and $\delta \in (0,1)$, letting $b_t := 2 \log \left( \frac{n m \pi^2 t^2}{3 \delta} \right)$, the M-GP-LUCB algorithm applied to $\Gamma(K)$ with $K$ at least $K_\epsilon \coloneqq \left\lceil \frac{b}{2\epsilon} \sqrt{\log \left( \frac{4a}{\delta}\right)} \, \right\rceil + 1$ returns a strategy profile $\overline{\bpi} \coloneqq (\bar{x}, \bar{y})$ such that $\mathbb{P} \left( |u(\bpi^\ast) - u(\bar{x}, y^\ast(\bar{x}))| \leq \epsilon \right) \geq  1- \delta$.
	Moreover, the algorithm stops its execution after at most:
	\begin{align}
		T_{\delta, \epsilon} \leq 64 \, H^*(u) \, \lambda \left[ \log \left( 64 \, H^*(u)\, \lambda \, \pi \, K_\epsilon \, \sqrt{\frac{1}{3 \delta}} \right) + \right. \nonumber\\
	\left. + 2 \, \log \left( \log \left( 64 \, H^*(u) \lambda \, \pi \, K_\epsilon \sqrt{\frac{1}{3 \delta}} \right) \right) \right],
	\end{align}
	where we require that $64 \, \lambda \, \pi \, K_\epsilon \, \sqrt{\frac{1}{3 \delta}} > 4.85$.
\end{restatable}

\begin{restatable}{corollary}{corse}\label{cor:cont_budget}
	Assume that $u$ is drawn from a GP$(\mathbf{0}, k(\bpi,\bpi'))$ satisfying Assumption~\ref{ass:smooth}.
	Given $\epsilon > 0$ and $\delta \in (0,1)$, letting $T$ be the number of available rounds, the GP-SE algorithm applied to $\Gamma(K)$ with $K$ at least $K_\epsilon \coloneqq \left\lceil \frac{b}{2\epsilon} \sqrt{\log \left( \frac{4a}{\delta}\right)} \, \right\rceil + 1$ returns a profile $\overline{\bpi} \coloneqq (\bar{x}, \bar{y})$ such that $\mathbb{P} \left( |u(\bpi^\ast) - u(\bar{x}, y^\ast(\bar{x}))| > \epsilon \right) < \delta_{T, \epsilon}$, where:
	\begin{equation}
		\delta_{T,\epsilon} \coloneqq 4 K^2_\epsilon (K_\epsilon - 1) e^{-\frac{T - K^2_\epsilon}{8 \lambda \overline{\log}(K^2_\epsilon) H_2}} + 2 a e^{-\frac{b^2}{4 \epsilon^2 (K_\epsilon - 1)^2}}.
	\end{equation}
\end{restatable}

In the result of Corollary~\ref{cor:cont_budget}, the discretization parameter $K_\epsilon$ depends on a confidence level $\delta$ that has to be chosen in advance.
Another possibility is to try to minimize the overall confidence $\delta_{T, \epsilon}$ by appropriately tuning the parameter $\delta$.
Formally, a valid confidence level can be defined as follows:
\begin{equation}
	\delta_{\textnormal{opt}} \coloneqq \inf \left\{ \delta \in (0,1) \ : \  \delta_{T,\epsilon} \right\},
\end{equation}
noticing that $\delta_{T, \epsilon}$ depends on $\delta$ also through the term $K_\epsilon$.
Unfortunately, this minimization problem does not admit a closed-form optimal solution.
Nevertheless, we can compute an (approximate) optimal value for $\delta$ by employing numerical methods~\cite{nocedal2006numerical}.

\begin{table*}[th!]
	\caption{Experimental results of algorithms M-LUCB, M-G-LUCB, and M-GP-LUCB on SBGs with finite strategy spaces.}
	\label{tab:ressqua}
	\centering
	\begin{tabular}{ccc|ccc|ccc|ccc|c}
		&                                               &      & \multicolumn{3}{c|}{M-LUCB} & \multicolumn{3}{c|}{M-G-LUCB} & \multicolumn{3}{c|}{M-GP-LUCB} & GP-SE \\ \cline{4-13} 
		\multicolumn{1}{l}{}                            & \multicolumn{1}{l|}{}                         & $T$    & $T_\delta$  & $\%end$ & $\%opt$ & $T_\delta$  & $\%end$  & $\%opt$ & $T_\delta$  & $\%end$  & $\%opt$  & $\%opt$ \\ \hline
		\multicolumn{1}{c|}{\multirow{3}{*}{SQE}}       & \multicolumn{1}{c|}{$l = 0.1$}                & 30k  & 10673.86    & 53.33 & 87.13 & ~~227.23      & \hspace{0.18cm}96.70   & 86.73 & \hspace{0.18cm}229.19      & \hspace{0.18cm}93.33  & \hspace{0.18cm}86.66  & 100.00   \\ \cline{2-13} 
		\multicolumn{1}{c|}{}                           & \multicolumn{1}{c|}{\multirow{2}{*}{$l = 2.0$}} & 30k  & 23788.96    & 13.73 & 84.80  & 2460.19     & \hspace{0.18cm}89.23   & 77.73 & 2020.89 & \hspace{0.18cm}76.66  & \hspace{0.18cm}93.36  & \hspace{0.18cm}93.23 \\
		\multicolumn{1}{c|}{}                           & \multicolumn{1}{c|}{}                         & 100k & 42103.86    & 46.66 & 88.63 & 3535.00        & \hspace{0.18cm}91.86   & 78.56 & 5656.77      & \hspace{0.18cm}76.77  & \hspace{0.18cm}93.30   & \hspace{0.18cm}96.60  \\ \hline
		\multicolumn{1}{c|}{\multirow{3}{*}{$M_{1.5}$}} & \multicolumn{1}{c|}{$l = 0.1$}                & 30k  & 13869.59    & 56.06 & 66.66 & \hspace{0.18cm}222.06      & 100.00  & 66.83 & \hspace{0.18cm}224.87      & 100.00    & \hspace{0.18cm}66.66  & 100.00   \\ \cline{2-13} 
		\multicolumn{1}{c|}{}                           & \multicolumn{1}{c|}{\multirow{2}{*}{$l = 2.0$}} & 30k  & 18978.75    & 33.33 & 76.26 & 2532.98     & \hspace{0.18cm}91.30   & 76.90 & 3775.65     & \hspace{0.18cm}88.03  & \hspace{0.18cm}78.86  & \hspace{0.18cm}98.30  \\
		\multicolumn{1}{c|}{}                           & \multicolumn{1}{c|}{}                         & 100k & 28798.42    & 50.00 & 80.70  & 3662.00       & \hspace{0.18cm}95.36   & 77.00 & 4618.27	   & \hspace{0.18cm}93.33  & \hspace{0.18cm}79.53  & \hspace{0.18cm}98.76 \\ \hline
		\multicolumn{1}{c|}{\multirow{3}{*}{$M_{2.5}$}} & \multicolumn{1}{c|}{$l = 0.1$}                & 30k  & 13335.26    & 79.80  & 86.66 & \hspace{0.18cm}168.61      & \hspace{0.18cm}97.90   & 86.06 & \hspace{0.18cm}171.62      & \hspace{0.18cm}96.66  & \hspace{0.18cm}86.66  & \hspace{0.18cm}99.83 \\ \cline{2-13} 
		\multicolumn{1}{c|}{}                           & \multicolumn{1}{c|}{\multirow{2}{*}{$l = 2.0$}} & 30k  & 20404.41    & 24.66 & 89.26 & 1984.55     & \hspace{0.18cm}92.10   & 86.63 & 2435.84     & \hspace{0.18cm}88.24  & \hspace{0.18cm}95.27  & \hspace{0.18cm}95.60  \\
		\multicolumn{1}{c|}{}                           & \multicolumn{1}{c|}{}                         & 100k & 49198.82    & 62.06 & 92.86 & 2617.31     & \hspace{0.18cm}93.70   & 87.13 & 2626.89 & \hspace{0.18cm}86.66     & \hspace{0.18cm}94.93    & \hspace{0.18cm}96.46
	\end{tabular}
\end{table*}

\subsection{Arbitrary Discretization}

Whenever using an equally-spaced grid as a discretization scheme is unfeasible, the theoretical results based on Theorem~\ref{thm:cont} do not hold anymore.
Nevertheless, given any finite sets of players' strategies, we can bound with high probability the distance of a maximin profile $\bpi^\ast$ from the strategy profile learned in the resulting (finite) discretized SBG.
Formally, let $\mathcal{X}_n \subseteq \X$ be a finite set of $n$ first player's strategies and, similarly, let $\mathcal{Y}_m \subseteq \Y$ be a finite set of $m$ second player's strategies.
Thus, the resulting finite SBG $\Gamma \coloneqq (\X_n, \Y_m, u)$ has $n \, m$ strategy profiles.
Let
\begin{align*}
	& d_x^{\max} \coloneqq \max_{x \in \mathcal{X}} \min_{x_i \in \mathcal{X}_n} |x - x_i|,\\
	& d_y^{\max} \coloneqq \max_{y \in \mathcal{Y}} \min_{y_i \in \mathcal{Y}_m} |y - y_i|,
\end{align*}
then, we can show the following result.
\begin{restatable}{theorem}{thmsupercazzola}\label{thm:supercazzola}
	Assume that $u$ is drawn from a \emph{GP}$(\mathbf{0},k(\bpi,\bpi'))$ satisfying Assumption~\ref{ass:smooth}.
	Given $\delta \in (0,2)$, let $\overline{\bpi} \coloneqq (\bar{x}, \bar{y}) \in \mathcal{X}_n \times \mathcal{Y}_m$ be a maximin strategy profile for a finite game $\Gamma \coloneqq (\mathcal{X}_n, \mathcal{Y}_m, u)$.
	Then, the following holds:
	\begin{equation*}
	\mathbb{P} \left( \left| u(\bpi^\ast) - u(\overline{\bpi}) \right| \leq b \sqrt{\log \left( \frac{4a}{\delta} \right)} \max \left\{d_x^{\max}, d_y^{\max}\right\}\right) \geq 1 - \frac{\delta}{2}.
	\end{equation*}
\end{restatable}

Let us remark that the result in Theorem~\ref{thm:supercazzola} can be applied any time using an equally-spaced grid as a discretization scheme is unfeasible, as it is the case, \emph{e.g.}, when the dimensionality $d$ of the players' strategy spaces is too large.

\section{Experimental Results}\label{sec:experiments}

We experimentally evaluate our algorithms on both finite and infinite SBGs.
As for the finite case, we compare the performances (with different metrics) of our M-GP-LUCB and GP-SE algorithms against two baselines.
The first one is the M-LUCB algorithm proposed by~\citet{garivier2016maximin}, which is the state of the art for learning maximin strategies in finite SBGs and can be easily adapted to our setting by using a different exploration term $b_t$.~\footnote{Since our utilities are not in $[0,1]$ (as they are drawn from a Gaussian instead of a Bernoulli), we  multiply the $b_t$ provided in~\cite{garivier2016maximin} by the utility range.} 
We introduce a second baseline to empirically evaluate how our algorithms speed up their convergence by leveraging correlation of the utilities.  
Specifically, it is a variation of our M-GP-LUCB algorithm (called M-G-LUCB) where utility values are assumed drawn from independent Gaussian random variables, instead of a GP.~\footnote{The formulas for updating the mean $\mu_t$ and the variance $\sigma^2_t$ of the posterior distribution are changed accordingly.}
As for SBGs with infinite strategy spaces, there are no state-of-the-art techniques that we can use as a baseline for comparison.
Thus, we show the quality (in terms of $\epsilon$) of the strategy profiles returned by our algorithms using different values of $K_\epsilon$ for the discretized games.
The average $\epsilon$ values obtained empirically (called $\hat{\epsilon}$ thereafter) are compared against the theoretical values prescribed by Theorem~\ref{thm:cont} (for the given $K_\epsilon$), so as to evaluate whether our bounds are strict or not.

\subsection{Random Game Instances}
As for finite SBGs, we test the algorithms on random instances generated by sampling from GP$(\mathbf{0}, k(\bpi, \bpi'))$, using the following two commonly used kernel functions (see~\cite{williams2006gaussian} for more details):
\begin{itemize}
	\item \emph{squared exponential}:
	$k(\bpi,\bpi') \coloneqq e^{- \frac{1}{2 l^2} || \bpi - \bpi' ||^2},$
	where $l$ is a length-scale parameter;
	\item \emph{Mat\'ern}:
	$k(\bpi,\bpi') \coloneqq \frac{2^{1 - \nu}}{G(\nu)} r^\nu B_\nu(r),$
	where $r \coloneqq \frac{\sqrt{2 \nu}}{l} || \bpi - \bpi' ||$, $\nu$ controls the smoothness of the functions, $l$ is a length-scale parameter, $B_\nu$ is the second-kind Bessel function, and $G$ is the Gamma function.
\end{itemize}
We set the kernel parameters to $l \in \{0.1, 2\}$ and $\nu \in \{1.5, 2.5\}$, generating $30$ instances for each possible combination of kernel function and parameter values.
As for SBGs with infinite strategy spaces, we test on instances generated from distributions with $l = 0.1$ and, with the Mat\'ern kernel, $\nu \in \{1.5, 2.5\}$.
The infinite strategy spaces are approximated with a discretization scheme based on a grid made of $100$ equally-spaced points.

In the fixed-confidence setting, we let $\delta = 0.1$ and stop the algorithms after $T \in \{30\textnormal{k}, 100\textnormal{k}\}$ rounds.
Similarly, the GP-SE algorithm is run with a budget $T \in \{30\textnormal{k}, 100\textnormal{k}\}$.
For each possible combination of algorithm, game instance, and round-limit $T$, we average the results over $100$ runs.

\paragraph{Results on Finite SBGs}
The results are reported in Table~\ref{tab:ressqua}, where $T_\delta$ is the average number of queries used by the algorithm in the runs not exceeding the round-limit $T$, $\%end$ is the percentage of runs the algorithm terminates before $T$ rounds, and $\%opt$ is the percentage of runs the algorithm is able to correctly identify the maximin profile $\bpi^*$.
Notice that M-GP-LUCB and M-G-LUCB clearly outperform M-LUCB, as the latter requires a number of rounds $T_\delta$ an order of magnitude larger.
M-GP-LUCB and M-G-LUCB provide similar performances in terms of $T_\delta$, but the former identifies the maximin profile more frequently than the latter.
While always using the maximum number of rounds $T$, GP-SE is the best algorithm in identifying the maximin profile.

\paragraph{Results on Infinite SBGs}
Figure~\ref{fig:cont} provides the values of $\epsilon$ and $\hat{\epsilon}$ for an instance generated from a Mat\'ern kernel with $\nu = 2.5$ (see Appendix~E for more results).
%
In all the instances, $\hat{\epsilon}$ is lower than $\epsilon$, empirically proving the correctness of the guarantees provided in Section~\ref{sec:continuous}.
Moreover, as expected, $\hat{\epsilon}$ decreases as the number of discretization points $K_\epsilon$ increases.

\subsection{Security Game Instances}

We also test on a SBG instance with infinite strategy spaces inspired by the real-world security game setting described in Section~\ref{sec:intro}.
This game models a military scenario in which a terrestrial counter-air defensive unit has to fire a heat-seeking missile to an approaching enemy airplane, which, after the missile has been launched, can deploy an obfuscating flare so as to try to deflect it.
We call this game \emph{Hit-the-Spitfire}.
The model underlying such game and the parameters used in the experiment are depicted in Figure~\ref{fig:gameinstance}, where $h_\perp$ is the distance between the airplane and the terrestrial unit, $h_f$ is the distance of the flare from the plane, $v_a$ and $v_d$ are the speed of the missile and the plane, respectively, while $\ell$ is the length of the plane, with the flare covering half of this space ($\frac{\ell}{2}$).
The first player (the counter-air defensive unit) can determine the angle $\theta \in [0, 1]$ (radians) at which the missile is launched, while the second player (the airplane) has to decide the position $s \in [0, s_{\max}]$ where to release the flare.
If the missile hits the plane, then it incurs damage $d \in \mathbb{R}^+$ that depends on the hitting point (the nearer to the center of the plane, the higher).
If the missile hits the flare, then there is some probability that it is deflected away from the airplane, otherwise, the missile still hits the target.
The probability of deflection is large when the distance of the airplane from the deployed flare is larger.~\footnote{We provide the complete description of the setting in Appendix~D.}
We run the M-GP-LUCB with $\delta=0.1$.

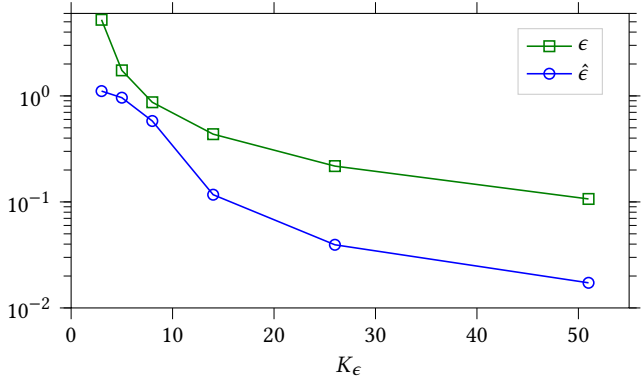
\begin{figure}[t!]
	\begin{tikzpicture}

\definecolor{color0}{rgb}{0.75,0.75,0}

\begin{axis}[
width=9cm, height=5.5cm,
legend cell align={left},
legend style={at={(0.80,0.95)}, anchor=north west, draw=white!80.0!black},
tick align=outside,
tick pos=both,
x grid style={lightgray!92.02614379084967!black},
xlabel={$K_\epsilon$},
xmin=0, xmax=55,
xtick style={color=black},
y grid style={lightgray!92.02614379084967!black},
ymin=0.01, ymax=6,
ytick style={color=black},
ymode=log
]

\addplot [semithick, green!50.0!black, mark=square]
table {%
3	5.226517923
5	1.742172641
8	0.871086321
14	0.43554316
26	0.21777158
51	0.106663631	
};
\addlegendentry{$\epsilon$}
\addplot [semithick, blue, mark=o]
table {%
3	1.110273066
5	0.962151103
8	0.579804642
14	0.116961069
26	0.039421522
51	0.017258231
};
\addlegendentry{$\hat{\epsilon}$}

\end{axis}

\end{tikzpicture}
	\caption{$\epsilon$ vs. $\hat{\epsilon}$ for different values of $K_\epsilon$ (Mat\'ern kernel with smoothness parameter $\nu = 2.5$).}
	\label{fig:cont}
\end{figure}

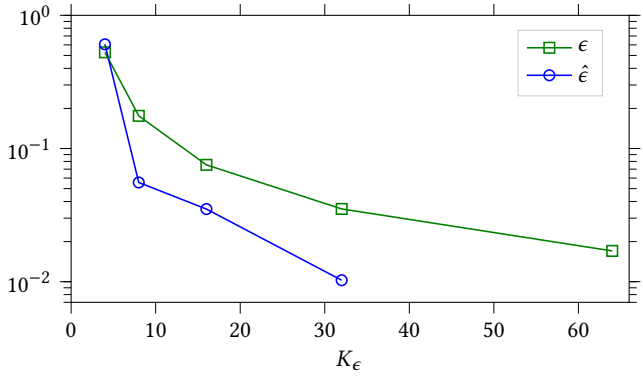
\begin{figure}[t!]
	\begin{tikzpicture}

\definecolor{color0}{rgb}{0.75,0.75,0}

\begin{axis}[
width=9cm, height=5.4cm,
legend cell align={left},
legend style={at={(0.8,0.95)}, anchor=north west, draw=white!80.0!black},
tick align=outside,
tick pos=both,
x grid style={lightgray!92.02614379084967!black},
xlabel={$K_\epsilon$},
xmin=0, xmax=66,
xtick style={color=black},
y grid style={lightgray!92.02614379084967!black},
ymin=0.007,
ymax=1,
ytick style={color=black},
ymode=log
]

\addplot [semithick, green!50.0!black, mark=square]
table {%
4	0.527812353
8	0.175937451
16	0.075401765
32	0.035195889
64	0.017030269
};
\addlegendentry{$\epsilon$}
\addplot [semithick, blue, mark=o]
table {%
4	0.605900626
8	0.05560656
16	0.035119471
32	0.010256103
};
\addlegendentry{$\hat{\epsilon}$}

\end{axis}

\end{tikzpicture}
	\caption{$\epsilon$ vs. $\hat{\epsilon}$ for different values of $K_\epsilon$ (\emph{Hit-the-Spitfire} security game).}
	\label{fig:missile}
\end{figure}

\begin{figure}[t!]
	\includegraphics[width=0.4 \textwidth]{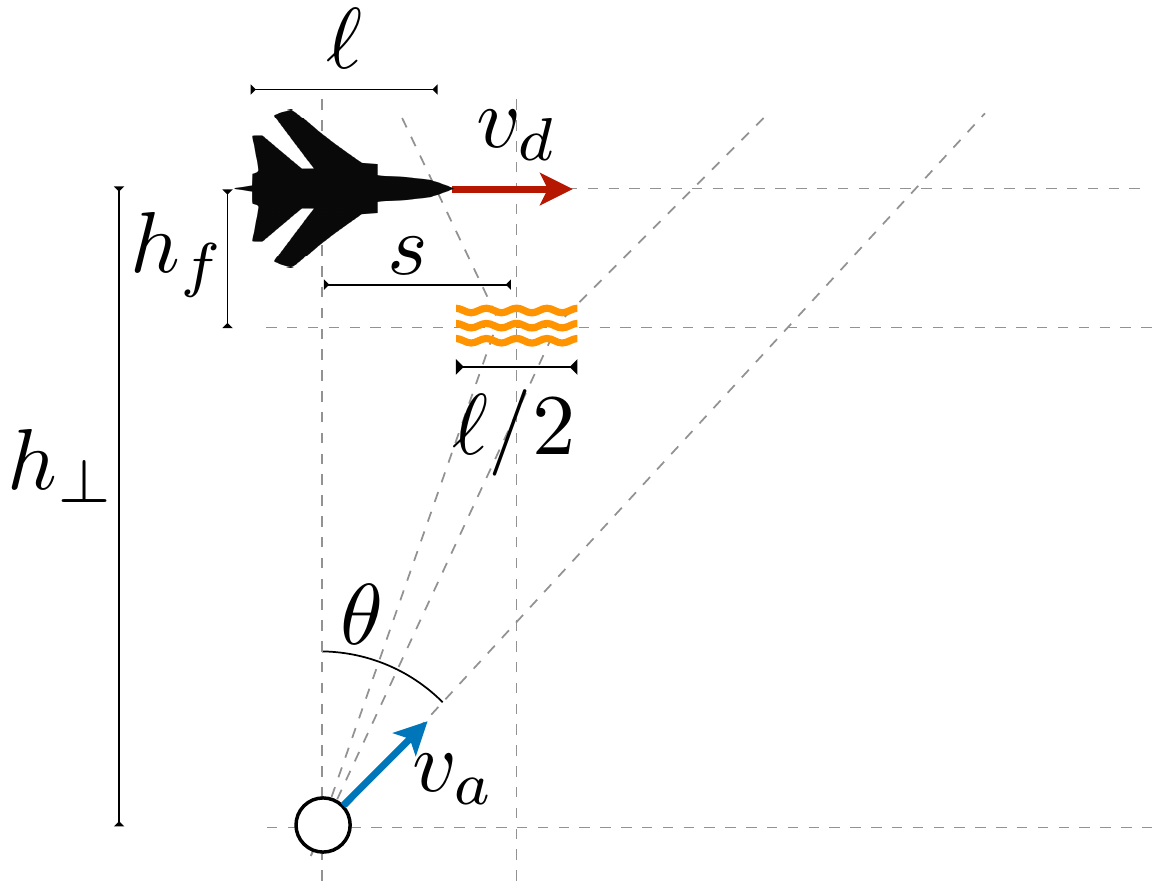}
	\put (-357, 50){\makebox[0.7\textwidth][r]{
			\small
			\begin{tabular}{|c|c|}
				\hline
				Parameter & Value\\
				\hline
				$h_\perp$& $100$ m \\
				\hline
				$h_f$ & $10$ m \\
				\hline
				$v_a$ & $500$ m/s \\
				\hline
				$v_d$ & $120$ m/s \\
				\hline
				$\ell$ & $15$ m \\
				\hline
	\end{tabular}}}
	\caption{\emph{Hit-the-Spitfire} security game instance and values for its parameters used in the experiments.}
	\label{fig:gameinstance}
\end{figure}

\paragraph{Results}
Figure~\ref{fig:missile} reports the results of running the M-GP-LUCB algorithm with $\delta=0.1$ on the \emph{Hit-the-Spitfire} game (performing $100$ runs for each $K_\epsilon$).
Notice that, in most of the cases, $\hat{\epsilon}$ is lower than the theoretical value $\epsilon$.
This is unexpected, since, in this setting, the assumption that the utility function $u$ is drawn from a GP does not hold.
We remark that, in all the runs, M-GP-LUCB is able to identify the maximin strategy profile over the given grid.

\section{Discussion and Future Works}

We addressed the problem of learning \emph{maximin} strategies in two-player zero-sum SBGs with \emph{infinite strategy spaces}, providing algorithms with theoretical guarantees.
To the best of our knowledge, we provided the first learning algorithms for infinite SBGs enjoying $\delta$-PAC theoretical guarantees on the quality of the returned solutions.
This significantly advances the current state of the art for SBGs, as dealing with infinite strategies paves the way to the application of such models in complex real-world settings.
The fundamental ingredient of our results is the assumption that the utility functions are drawn from a GP, which allows us to encode function regularities without relying on specific parametric assumptions, such as, \emph{e.g.}, linearity. 
%
%

In future, we will extend our work along different directions.
For instance, we may address the case of general (\emph{i.e.}, non-zero-sum and with more than two players) SBGs with finite (or even infinite) strategy spaces, where one seeks for an (approximate) Nash equilibrium.
Along this line, an interesting question is how to generalize our learning algorithms based on best arm identification techniques to deal with Nash-equilibrium conditions instead of maximin ones.
This would pave the way to the application of our techniques to other interesting problems, such as multi-agent evaluation by means of meta-games~\cite{tuyls2018generalised,rowland2019multiagent}.
Another interesting direction for future works is to study how to apply our techniques in empirical mechanism design problems~\cite{viqueira2019empirical}.

\section*{Acknowledgments}

This work has been partially supported by the Italian MIUR PRIN 2017 Project ALGADIMAR ``Algorithms, Games, and Digital Market''.

\clearpage
\bibliographystyle{ACM-Reference-Format}
\bibliography{biblio}

\clearpage
\appendix
\onecolumn

\section{Omitted Proofs for the Fixed-Confidence Setting}\label{sec:appendix_fixed_confidence}

We provide the complete proof of Theorem~\ref{thm:fixed_confidence}, which shows that the M-GP-LUCB algorithm (Algorithm~\ref{alg:m_gplucb}) proposed in the fixed-confidence setting is $\delta$-PAC and stops its execution after at most $T_\delta$ rounds, as a function of the confidence level $\delta$.

Let us start by recalling a result needed for the proof of Theorem~\ref{thm:fixed_confidence} and provided in~\cite[Lemma~5.1]{srinivas2009gaussian}, which we formally state using our notation in the following. 

\begin{lemma}\label{lem:delta}
	Given $\delta \in (0,1)$, let $b_t \coloneqq 2 \log \left( \frac{n m \pi^2 t^2}{6 \delta} \right)$. Then,
	\begin{equation}
	\mathbb{P} \left( \forall \bpi \in \Pi \quad \forall t \geq 1 \quad |u(\bpi) - \mu_t(\bpi)| \leq \sqrt{b_t} \sigma_t(\bpi) \right) > 1 - \delta.
	\end{equation}
\end{lemma}

The proof of Theorem~\ref{thm:fixed_confidence} also needs the following ancillary result, which, assuming a prior distribution {GP}$(\mathbf{0}, k(\bpi,\bpi'))$ over $u$, provides an upper bound on the variance of the posterior distribution at round $t$.

\begin{lemma}\label{lem:variance}
	For every strategy profile $\bpi \in \Pi$, the variance $\sigma_t^2(\bpi)$ is upper bounded as follows:
	\begin{equation}
	\sigma^2_t(\bpi) \leq \frac{\lambda}{\frac{\lambda}{\sigma^2} + N_t(\bpi)},
	\end{equation}
	where $N_t(\bpi)$ denotes the number of rounds in which $\bpi$ has been selected by the algorithm, up to round $t$.
\end{lemma}
\begin{proof}
	Since the estimated variance $\sigma_t^2(\bpi)$ is nonincreasing in $t$~\cite{williams2006gaussian} and the order in which the algorithm selects the strategy profiles does not influence the estimates, we can upper bound $\sigma_t^2(\bpi)$ with an estimate obtained using a specific subset of the  observations.
	Here, we only consider those rounds in which $\bpi$ has been selected.
	This way, thanks to the Sherman-Morrison inversion formula~\cite{petersen2008matrix}, we have a closed-form expression for the estimated variance after the strategy profile $\bpi$ has been selected $N_t(\bpi)$ times.
	Thus,
	\begin{equation}
	\sigma_t^2(\bpi) \leq \sigma^2(N_t(\bpi)) = \frac{\lambda}{\frac{\lambda}{\sigma^2} + N_t(\bpi)}.
	\end{equation}
\end{proof}

\thmconfidence*

\begin{proof}
	We adapt the proof for the M-LUCB algorithm provided in~\cite{garivier2016maximin}, in which a similar result has been derived for a setting where the utility values are drawn from distributions with finite support and these distributions are independent (in contrast, in our setting, the distributions do not have finite support and they are not independent, since we assume that the utility function is drawn from a GP).

	Let $E$ be the event in which all the real utility values $u(\bpi)$, for each $\bpi \in \Pi$, are contained in the intervals $[L_t(\bpi), U_t(\bpi)]$ for every round $t \in \mathbb{N}$.
	Next, we prove the correctness of the algorithm under the assumption that the event $E$ holds.

	Since $E$ holds, for every strategy of the first player $x \in \X$, the second player's best response utility $u(x, y^\ast(x))$ cannot be below the lowest among the lower bounds $L_t(x, y)$, for $y \in \Y$, for every round $t$.
	Consequently, this is also true for round $T_\delta$.
	At round $t = T_\delta$, given how the stopping rule is defined, we have that the confidence interval of the recommended strategy profile $\bar{\bpi} = (\bar{x}_t, c_t(\bar{x}_t))$ is disjoint with all the intervals of strategy profiles in $\{ (x, c_t(x)) \}_{x \neq \bar{x}_t }$.
	Therefore,
	\begin{align}
		\overline{u} \coloneqq u(\bar{\bpi}) & \geq \mu_t(\bar{\bpi}) - \sqrt{b_t} \sigma_t(\bar{\bpi}) > \max_{x \in \mathcal{X} \setminus \{ \bar{x}_t \}} \left( \mu_t(x, y^\ast(x)) + \sqrt{b_t} \sigma_t(x, y^\ast(x)) \right) \nonumber\\
	& \geq \max_{x \in \mathcal{X} \setminus \{ \bar{x}_t \}} u(x, y^\ast(x)) \coloneqq \underline{u},
	\end{align}
	where we used the fact that all the utilities are contained in the confidence intervals, since $E$ holds.

	Now, we focus on the $\delta$- PAC properties of the algorithm.
	Since Lemma~10 in~\cite{garivier2016maximin} is still valid, we have that, given $c \in [\underline{u}, \overline{u}]$ and $t < T_\delta$, if $E$ holds, then there exists a strategy profile $\bpi=(x,y) \in \{\bpi_{t+1}, \bpi_{t+2}\}$ such that:
	\begin{equation}\label{eq:lemma_10}
		\mu_t(\bpi) - \sqrt{b_t} \sigma_t(\bpi) \leq c \leq \mu_t(\bpi) + \sqrt{b_t} \sigma_t(\bpi).
	\end{equation}
	Let us consider the case in which $\bpi \neq (x, y^\ast(x))$.
	We have that the utility $u(x,y^\ast(x))$ of a best response $y^\ast(x)$ to $x$ is contained in the confidence interval for $\bpi$.
	Indeed, since the algorithm selected the strategy profile $\bpi$ (either at round $t+1$ or at round $t+2$), it holds:
	%
	%
	\begin{equation}\label{eq:ineq_1}
		\mu_t(\bpi) - \sqrt{b_t} \sigma_t(\bpi) \leq \mu_t(x,y^\ast(x)) - \sqrt{b_t} \sigma_t(x,y^\ast(x)) \leq u(x,y^\ast(x)).
	\end{equation}
	Moreover, since the event $E$ holds we have that:
	\begin{equation}\label{eq:ineq_2}
		\mu_t(\bpi) + \sqrt{b_t} \sigma_t(\bpi) \geq u(\bpi)\geq u(x,y^\ast(x)).
	\end{equation}
	Equations~\eqref{eq:lemma_10},~\eqref{eq:ineq_1},~and~\eqref{eq:ineq_2} imply that $c$ and $u(x,y^\ast(x))$ are no further than $2 \sqrt{b_t} \sigma_t(\bpi)$, the length of the interval. 
	Formally:
	\begin{equation} \label{eq:diameter}
		|c - u(x,y^\ast(x))| \leq 2 \sqrt{b_t} \sigma_t(\bpi).
	\end{equation}
	%
	%
	%
	%
	%
	Recalling that $N_t(\bpi)$ denotes the number of rounds in which $\bpi$ has been selected by the algorithm up to round $t$, Lemma~\ref{lem:variance} tells us that $\sigma_t^2(\bpi) \leq \sigma^2(N_t(\bpi))$, which, together with Equation~\eqref{eq:diameter}, allows us to show the following:
	\begin{equation}
		|c - u(x,y^\ast(x))| \leq 2 \sqrt{ \frac{\lambda b_t}{\frac{\lambda}{\sigma^2} + N_t(\bpi)}}.
	\end{equation}
	Solving for $N_t(\bpi)$ we get:
	\begin{equation}\label{eq:ineq_enne_1}
		N_t(\bpi) \leq \frac{4 b_t \lambda}{(c - u(x,y^\ast(x)))^2} - \frac{\lambda}{\sigma^2}.
	\end{equation}
	Given that $E$ holds, we also have that $u(x,y^\ast(x))$ is in the confidence interval for $\bpi$, and, thus, its distance from $u(\bpi)$ is smaller than $2 \sqrt{b_t} \sigma_t(\bpi)$.
	Formally:
	\begin{equation} \label{eq:diameter2}
		|u(\bpi) - u(x,y^\ast(x))| \leq 2 \sqrt{b_t} \sigma_t(\bpi).
	\end{equation}
	Using similar arguments as before we have that:
	\begin{equation}\label{eq:ineq_enne_2}
		N_t(\bpi) \leq \frac{4 b_t \lambda}{(u(\bpi) - u(x,y^\ast(x)))^2} - \frac{\lambda}{\sigma^2}.
	\end{equation}
	Finally, since both Equation~\eqref{eq:ineq_enne_1} and Equation~\eqref{eq:ineq_enne_2} should hold at the same time, we have that:
	\begin{equation}
		N_t(\bpi) \leq 4 b_t \lambda \frac{1}{\max \Big\{ (u(\bpi) - u(x,y^\ast(x)))^2, (c - u(x,y^\ast(x)))^2 \Big\} } - \frac{\lambda}{\sigma^2}.
	\end{equation}
	
	The case in which $\bpi \equiv (x,y^\ast(x))$ can be treated analogously, since all the results on $c \in [\underline{u}, \overline{u}]$ are still valid, while the condition in Equation~\eqref{eq:diameter2} is always satisfied.
	Therefore, in this case we have:
	\begin{equation}
		N_t(x,y^\ast(x)) \leq \frac{4 b_t \lambda}{(c - u(x,y^\ast(x)))^2} - \frac{\lambda}{\sigma^2}.
	\end{equation}
	%
	
	Now, let us choose $c \coloneqq \frac{u(x^\ast, y^\ast) + u(x^\ast, y^{\ast \ast})}{2}$, where $x^\ast$ denotes a first player's maximin strategy and, for notational convenience, $x^{\ast\ast} \in \argmin_{x \in \X \setminus \{x^\ast\}} u(x, y^\ast(x))$ is a first player's maximin strategy when $x^\ast$ is removed.
	Then, let us define $H^*(u) := \sum_{\bpi \in \Pi} c(\bpi)$, where, for every strategy profile $\bpi = (x, y) \in \Pi$, $c(\bpi)$ is defined as follows:
	\begin{equation}\label{eq:c_epsilon}
	c(\bpi) \coloneqq \frac{1}{\max \left\{(u(\bpi) - u(x,y^\ast(x)))^2, \left(\frac{u(x^\ast, y^\ast(x^\ast)) + u(x^{\ast\ast}, y^\ast(x^{\ast\ast}))}{2} - u(x,y^\ast(x)) \right)^2 \right\} }.
	\end{equation}
	Given $t \in \mathbb{N}$ such that $t > 8 H^*(u) b_t \lambda - \frac{\lambda nm}{\sigma^2}$ we have that:
	\begin{align}
		\min \{T_\delta, t\} & = \sum_{h = 1}^{t} \mathbbm{1} \{h < T_\delta \} = 2 \sum_{h \in 2 \mathbb{N}, h < t} \mathbbm{1} \{h < T_\delta \} \\
		&  \leq 2 \sum_{h \in 2 \mathbb{N}, h < t} \mathbbm{1} \left\{\exists \bpi \in \Pi \text{ s.t. } N_h(\bpi) \leq 4 b_h \lambda c(\bpi) - \frac{\lambda}{\sigma^2} \right\}\\
		& \leq 2 \sum_{h \in 2 \mathbb{N}, h < t} \,\, \sum_{\bpi \in \Pi} \mathbbm{1} \bigg\{ (\bar{x}_h, \bar{y}_h) = \bpi \vee (\bar{x}_{h+1}, \bar{y}_{h+1}) = \bpi \bigg\} \mathbbm{1} \bigg\{ N_h(\bpi) \leq 4 b_t \lambda c(\bpi) - \frac{\lambda}{\sigma^2} \bigg\}\\
		& \leq 4 \sum_{\bpi \in \Pi} \left( b_t \lambda c(x_1, x_2) - \frac{\lambda}{\sigma^2} \right)\\
		& = 8 H^*(u) b_t \lambda - \frac{\lambda nm}{\sigma^2} < t,
	\end{align}
	where we used the fact that the algorithm takes a decision only on even rounds, and, thus, $T_\delta$ is even, and that $b_t$ is nondecreasing in $t$, which implies that $b_t \geq b_h$ for each $h < t$.
	The previous inequality implies that $T_\delta < t$, which ensures that the algorithm returns a strategy after at most $t$ rounds.
	Formally:
	\begin{equation}
		T_\delta \leq \inf \left\{t \in \mathbb{N} \ : \ 8 H^*(u) b_t \lambda - \frac{\lambda n m}{\sigma^2} < t \right\}.
	\end{equation}
	
	In conclusion, Lemma~\ref{lem:delta} provides an explicit formula for the upper bound on the number of rounds $T_\delta$, in the case in which we require confidence at least $1 - \delta$ that the event $E$ occurs.
	Using the prescribed $b_t \coloneqq  2 \log \left( \frac{nm \pi^2 t^2}{6 \delta} \right)$ we have:
	\begin{equation}
		\inf \left\{t \in \mathbb{N} \ : \ 16 H^*(u) \lambda \log \left( \frac{n m \pi^2 t^2}{6 \delta} \right) - \frac{\lambda n m}{\sigma^2} < t \right\} \leq \inf \left\{t \in \mathbb{N} \ : \ 32  H^*(u) \lambda \log \left( \frac{n m \pi^2 t^2}{6 \delta} \right) < t \right\}.
	\end{equation}
	Applying Lemma~12 in~\cite{kaufmann2016complexity} with $\alpha = 2$, $c_1 = \frac{1}{32 H^*(u) \lambda}$, $c_2 =  \frac{n m \pi^2}{6 \delta}$, and $x = t$ we get:
	\begin{equation}
	\tau_\delta \leq 64 H^*(u) \lambda \left( \log \left( 64 H^*(u) \lambda \pi \sqrt{\frac{n m}{6 \delta}} \right) + 2 \log \left( \log \left( 64 H^*(u) \lambda \pi \sqrt{\frac{n m}{6 \delta}} \right) \right) \right),
	\end{equation}
	which holds provided that $64 \lambda \pi \sqrt{\frac{n m}{6 \delta}} > 4.85$.
\end{proof}

\section{Omitted Proofs for the Fixed-Budget Setting}\label{sec:appendix_fixed_budget}

We provide the complete proof of Theorem~\ref{thm:fixed_budget}, which shows that the GP-SE algorithm (Algorithm~\ref{alg:gp_se}) with $T$ rounds available returns a maximin profile with probability at least $1 - \delta_T$.

\thmse*

\begin{proof}
	Let us focus on the probability that the algorithm excludes a maximin profile $\bpi^\ast = (x^\ast, y^\ast(x^\ast))$ during a specific phase $p \in \{1, \ldots P - 1\}$.
	There exist two different sources of error, which occur in a disjoint way.
	In the first case, the algorithm has already discarded all the strategy profiles in $\{ (x^\ast, y)  \}_{y \in \Y : (x^\ast, y) \neq \bpi^\ast }$, while, in the second one, there are some strategy profiles left in such set.
	Let us define the following events:
	\begin{align}
	E_p &= \bigg\{ {\Pi}_{p-1} \cap \{ \bpi^\ast \} = \{ \bpi^\ast \} \bigg\} \cap \bigg\{ {\Pi}_{p} \cap \{ \bpi^\ast \} = \emptyset \bigg\}, \\
	F &= \bigg\{ \forall y \neq y^\ast(x^\ast) \in \Y \quad {\Pi}_{p-1} \cap\{(x^\ast, y)\} = \emptyset \bigg\}.
	\end{align}
	
	Then, the probability that the algorithm makes an error during the learning process is:
	\begin{align}
	&\mathbb{P}(E) \coloneqq \sum_{k=1}^{{P - 1}} \mathbb{P}(E_p) =  \sum_{p=1}^{{P - 1}}  \mathbb{P}(E_p  \cap F) + \mathbb{P}(E_p \cap F^C) = \sum_{p=1}^{{P - 1}} \mathbb{P}(E_p \cap F^C) + \sum_{p = m-1}^{{P - 1}} \mathbb{P}(E_p  \cap F)\\
	& = \sum_{p=1}^{{P - 1}} \mathbb{P} \left(\exists y \neq y^\ast(x^\ast) \in \Y \text{ s.t. } \min_{y \in \Y} \mu_p(x^\ast, y) \leq \min_{\bpi \in \Pi_p} \mu_p(\bpi) \ \wedge \forall y \in \Y \ \mu_p(\bpi^\ast) \geq \mu_p(x^\ast, y) \right) + \nonumber \\
	& \quad+ \sum_{p = m-1}^{{P - 1}} \mathbb{P} \bigg( \forall x \in \X \setminus \{x^\ast \}, \forall y \in \Y \ \mu_p(\bpi^\ast) \leq \mu_p(x,y) \bigg)\\
	& \leq \sum_{p=1}^{{P - 1}} \mathbb{P} \bigg( \forall y \in \Y \ \mu_p(\bpi^\ast) \geq \mu_p(x^\ast, y) \bigg) + \sum_{p = m-1}^{{P - 1}} \mathbb{P} \bigg( \forall x \in \X \setminus \{x^\ast\} \ \mu_p(\bpi^\ast) \leq \mu_p(x, y^\ast(x)) \bigg)\\
	& \leq \sum_{p=1}^{{P - 1}} \sum_{y \in \Y \setminus \{y^\ast(x^\ast)\}} \mathbb{P} \bigg( \mu_p(\bpi^\ast) \geq \mu_p(x^\ast, y) \bigg) + \sum_{p = m-1}^{{P - 1}} \sum_{ x \in \X \setminus \{x^\ast\}} \mathbb{P} \bigg(\mu_p(\bpi^\ast) \leq \mu_p(x,y^\ast(x)) \bigg)\\
	& = \sum_{p=1}^{{P - 1}} \sum_{y \in \Y \setminus \{y^\ast(x^\ast)\}} \mathbb{P} \bigg( \mu_p(\bpi^\ast) - u(\bpi^\ast) - \frac{\Delta(x^\ast, y)}{2} - \mu_p(x^\ast, y) + u(x^\ast, y) - \frac{\Delta(x^\ast, y)}{2} \geq 0 \bigg) + \nonumber\\
	& \quad+ \sum_{p = m-1}^{P - 1} \sum_{x \in \X \setminus \{x^\ast\}} \mathbb{P} \bigg(\mu_p(\bpi^\ast) - u(\bpi^\ast) + \frac{\Delta(x,y^\ast(x))}{2} - \mu_p(x,y^\ast(x)) + u(x,y^\ast(x)) + \frac{\Delta(x,y^\ast(x))}{2} \leq 0 \bigg)  \label{eq:delta}\\
	& \leq \sum_{p=1}^{{P - 1}} \sum_{y \in \Y \setminus \{y^\ast(x^\ast)\}} \left( e^{- \frac{ \Delta^2(x^\ast, y)}{8 \sigma^2_p(\bpi^\ast)}} + e^{- \frac{\Delta^2(x^\ast, y)}{8 \sigma^2_p(x^\ast,y)}} \right) + \sum_{p = m-1}^{P - 1} \sum_{x \in \X \setminus \{x^\ast\}} \left( e^{-\frac{\Delta^2(x, y^\ast(x))}{8 \sigma^2_p(\bpi^\ast)}} + e^{-\frac{\Delta^2(x, y^\ast(x))}{8\sigma^2_p(x, y^\ast(x))}} \right)\\
	& \leq \sum_{p=1}^{{P - 1}} \sum_{y \in \Y \setminus \{y^\ast(x^\ast)\}} 2 e^{-\frac{ \left( \frac{\lambda}{\sigma^2} + T_p \right) \Delta^2(x^\ast, y)}{8 \lambda}} + \sum_{p = m-1}^{P - 1} \sum_{x \in \X \setminus \{x^\ast\}} 2 e^{-\frac{ \left( \frac{\lambda}{\sigma^2} + T_p \right) \Delta^2(x, y^\ast(x))}{8 \lambda}}\\
	& = 2 \left( \sum_{p=1}^{{P - 1}} \sum_{y \in \Y \setminus \{y^\ast(x^\ast)\}} e^{-\frac{T_p \Delta^2(x^\ast, y)}{8 \lambda}} + \sum_{p = m-1}^{P - 1} \sum_{x \in \X \setminus \{x^\ast\}} e^{-\frac{T_p \Delta^2(x, y^\ast(x))}{8 \lambda}}\right),
	\end{align}
	where in Equation~\eqref{eq:delta} we defined $\Delta(x^\ast, y) \coloneqq u(x^\ast, y) - u(\bpi^\ast)$ and $\Delta(x,y^\ast(x)) \coloneqq u(\bpi^\ast) - u(x,y^\ast(x))$.
	Note that the second summation is from $p = m - 1$ since, even in the worst case in which the algorithm first dismisses only strategy profiles $(x^\ast,y)$ for $y \neq y^\ast(x^\ast) \in \Y$, during the phases $p \in \{ 1, \ldots, m-1 \}$, there is still at least one profile $(x^\ast, y) \neq \bpi^\ast$ in $\Pi_p$.
	
	Using a phase length strategy as defined in~\cite{audibert2010best}, \emph{i.e.}, $T_p = \left\lceil \frac{T - P}{\overline{\log}(P) (P  + 1 - p)} \right\rceil$, where $\overline{\log}(P) = \frac{1}{2} + \sum_{i=2}^P \frac{1}{i}$, we obtain the following:
	\begin{align}
	\mathbb{P}(E) &\leq 2 \left( \sum_{p=1}^{{P - 1}} (m - 1) e^{-\frac{T-P}{8 \lambda \overline{\log}(P) H}} + \sum_{p = m-1}^{P - 1} (n - 1) e^{-\frac{T-P}{8 \lambda \overline{\log}(P) H}} \right)\\
	& \leq 2 P (n + m - 2) e^{-\frac{T-P}{8 \lambda \overline{\log}(P) H_2}},
	\end{align}
	where $H_2 \coloneqq \max_i i \Delta^{-2}_{(i)}$ and $\Delta_{(i)}$ is the $i$-th difference between a generic strategy profile $\bpi_i$ and the maximin one $\bpi^\ast = (x^\ast, y^\ast(x^\ast))$, or, formally, $\Delta_{(i)} \coloneqq |u(\bpi^*) - u(\bpi_i)|$ with $\Delta_{(1)} \leq \Delta_{(2)} \leq \ldots \leq \Delta_{(P)}$.
	
	This concludes the proof.
\end{proof}

\section{Omitted Proofs for Simulation-Based Games with Infinite Strategy Spaces}\label{sec:appendix_continuous}

We provide a complete proof of Theorem~\ref{thm:cont} and Corollaries~\ref{cor:cont_conf}~and~\ref{cor:cont_budget}.

\thmcont*

\begin{proof}
	Under Assumption~\ref{ass:smooth}, it is possible to show that the following holds:
	\begin{align}
		& |u(\bpi)- u(x',y)| \leq L |x - x'| \quad \forall \bpi = (x,y) \in \Pi, \forall x' \in \X \label{eq:smoothx}\\
		& |u(\bpi)- u(x,y')| \leq L |y - y'| \quad \forall \bpi = (x,y) \in \Pi, \forall y' \in \Y \label{eq:smoothy}
	\end{align}
	with probability $1 - 2ae^{-\frac{L^2}{b^2}}$, for any $L > 0$.
	
	We divide the analysis in two cases: depending on whether the maximin utility value $u(\bpi^\ast)$ for the original game with infinite strategy spaces is larger than the one for the finite discretized game, or not. 
	First, we focus on the former case.
	We show how to bound the difference $u(\bpi^\ast) - u(\overline{\bpi}) > 0$.
	Note that, with a discretization of $\Pi$ made of at least $h = \left\lceil \frac{b}{2\epsilon} \sqrt{\log \left( \frac{4a}{\delta}\right)} \right\rceil$ intervals per dimension, there exists at least one first player's strategy in the discretized game, namely $x_g$, such that the distance between a maximin strategy $x^\ast$ in the original game and the closest strategy on the grid $x_g$ is less than or equal to $\frac{h}{2}$.
	Notice that $y^\ast(x_g)$, \emph{i.e.}, a second player's best response in the original game to $x_g$, differs at most of $\frac{h}{2}$ from $y_g^\ast(x_g)$, \emph{i.e}., a second player's best response  to $x_g$ in the discretized game.
	Therefore, we have that:
	\begin{align}
	u(\bpi^\ast) - u(\overline{\bpi}) & = u(\bpi^\ast) - u(x^\ast,y_g^\ast(x_g)) + u(x^\ast,y_g^\ast(x_g)) - u(x_g,y_g^\ast(x_g)) + u(x_g,y_g^\ast(x_g)) - u(\overline{\bpi})\\
	& = \underbrace{\min_{y \in \Y} u(x^\ast,y) - u(x^\ast,y_g^\ast(x_g))}_{\leq 0} + \underbrace{u(x^\ast,y_g^\ast(x_g)) - u(x_g,y_g^\ast(x_g))}_{L |x^\ast - x_g|} + \underbrace{u(x_g,y_g^\ast(x_g)) - u(\overline{\bpi})}_{< 0}\\
	& \leq \frac{L}{2h},
	\end{align}
	where we used Equation~\eqref{eq:smoothx} and the fact that $u(\overline{\bpi})$ is larger than the utility value of any other best response of the second player and, specifically, also larger than $u(x_g,y_g^\ast(x_g))$.
	
	Now, let us consider the second case, \emph{i.e.}, the maximin utility value $u(\bpi^\ast)$ for the original game is smaller than the one  for the finite discretized game.
	We provide a bound on the difference $u(\overline{\bpi}) - u(\bpi^\ast) > 0$.
	Let us assume that $x^\ast$ is on the grid.
	Then, the distance of $\bpi^\ast$ from a second player's best response in the discretized game should be less than $\frac{1}{2h}$.
	Therefore, its value cannot be lower than $u(x_g,y_g^\ast(x_g)) - \frac{L}{2h}$ for every $x_g$ on the grid.
	From its definition we have that:
	\begin{align}
		& u(\bpi^\ast) > \max_{x_g} \left( u(x_g,y_g^\ast(x_g)) - \frac{L}{2h} \right) = \max_{x_g} u(x_g,y_g^\ast(x_g)) - \frac{L}{2h} = u(\overline{\bpi}) - \frac{L}{2h},
	\end{align}
	which implies that $u(\overline{\bpi}) - u(\bpi^\ast) \leq \frac{L}{2h}$.
	Since the maximum over a larger set cannot decrease, we have that this result holds also in the case $x^*$ is not on the grid.
	Overall we have that $|u(\bpi^\ast) - u(\overline{\bpi})| \leq \frac{L}{2h}$.

	Choosing $L = b \sqrt{\log \left( \frac{4a}{\delta}\right)}$, we have that:
	\begin{align}
	|u(\bpi^\ast) - u(\overline{\bpi})| \leq \frac{b \sqrt{\log \left( \frac{4a}{\delta}\right)}}{2 \left\lceil \frac{b}{2\epsilon} \sqrt{\log \left( \frac{4a}{\delta}\right)} \right\rceil} \leq \epsilon,
	\end{align}
	which holds with probability at least $1 - \frac{\delta}{2}$.
\end{proof}

\corconfidence*

\begin{proof}
	Using Theorem~\ref{thm:fixed_confidence} with $b_t := 2 \log \left( \frac{n m \pi^2 t^2}{3 \delta} \right)$ and Applying Lemma 12 in~\cite{kaufmann2016complexity} with $\alpha = 2$ $c_1 = \frac{1}{32 H^*(u) \lambda}$, $c_2 = \frac{n m \pi^2}{3 \delta}$ and $x = t$ we have that after at most:
	\begin{equation}
	T_\delta \leq 64 H^*(u) \lambda \left( \log \left( 64 H^*(u) \lambda \pi \sqrt{\frac{n m}{3 \delta}} \right) + 2 \log \left( \log \left( 64 H^*(u) \lambda \pi \sqrt{\frac{n m}{3 \delta}} \right) \right) \right)
	\end{equation}
	the M-GP-LUCB algorithm returns a maximin strategy profile for the finite discretized game with probability at least $1 - \frac{\delta}{2}$.
	Then, using Theorem~\ref{thm:cont}, we have an error with respect to a maximin profile of the original game of at most $\epsilon$, with a probability of at least $1 - \frac{\delta}{2}$.
	Overall, using a union bound, the statement of the theorem holds with probability at least $1 - \delta$.
	The number of rounds, with the chosen discretization, becomes:
	\begin{equation}
	\tau_\delta \leq 64 H^*(u) \lambda \left( \log \left( 64 H^*(u) \lambda \pi K_\epsilon \sqrt{\frac{1}{3 \delta}} \right) + 2 \log \left( \log \left( 64 H^*(u) \lambda \pi K_\epsilon \sqrt{\frac{1}{3 \delta}} \right) \right) \right),
	\end{equation}
	which concludes the proof.
\end{proof}

\corse*

\begin{proof}
	The proof is obtained by setting $\delta' = 2 \delta$ in Theorem~\ref{thm:cont}, which assures that the error with respect to a maximin strategy profile computed over $\Pi$ is at most $\epsilon$ with a probability at least $1 - \delta'$, and using Theorem~\ref{thm:fixed_budget} with $K = K^2_\epsilon$.
	As a result, we get:
	\begin{align}
		K_\epsilon - 1 = \left\lceil \frac{b}{2\epsilon} \sqrt{\log \left( \frac{2a}{\delta'}\right)} \right\rceil \implies \delta' \leq 2 a e^{-\frac{b^2}{4 \epsilon^2 (K_\epsilon - 1)^2}}.
	\end{align}
\end{proof}

\thmsupercazzola*

\begin{proof}
	The proof follows from the proof of Theorem~\ref{thm:cont}, by substituting $d_x^{\max}$ as the maximum distance between one of the available strategies of the first players and $x^* \in \mathcal{X}$ (instead of $\frac{1}{2h}$) and $d_y^{\max}$ as the maximum distance between one of the strategies of the second players and $y^* \in \mathcal{X}$ (instead of $\frac{1}{2h}$)
	
	Specifically, letting $x_g$ be the closest strategy in $\mathcal{X}_n$ to the maximin strategy $x^*$, if $u(\bpi^\ast) - u(\overline{\bpi}) > 0$, we have:
	$$u(\bpi^\ast) - u(\overline{\bpi}) \leq L |x^* - x_g | \leq L d_x^{\max}.$$
	
	Conversely, if $u(\bpi^\ast) - u(\overline{\bpi}) < 0$, and $x^* \in \mathcal{X}_n$ we have:
	\begin{equation}
		u(\bpi^\ast) > \max_{x_g} \left( u(x_g,y_g^\ast(x_g)) - L |y_g^*(x_g) - y^*| \right) \geq \max_{x_g} u(x_g,y_g^\ast(x_g)) - L d_y^{\max}  = u(\overline{\bpi}) - L d_y^{\max},
	\end{equation}
	and for the properties of the maximum we have that the result also holds if $x^* \not\in \mathcal{X}_n$.
	
	Finally, using the fact that $-L d_y^{\max} \leq u(\bpi^\ast) - u(\overline{\bpi}) \leq L d_x^{\max}$, we have that:
	\begin{equation}
		|u(\bpi^\ast) - u(\overline{\bpi})| \leq L \max \left\{d_x^{\max}, d_y^{\max}\right\} = b \sqrt{\log \left( \frac{4a}{\delta} \right)} \max \left\{d_x^{\max}, d_y^{\max}\right\},
	\end{equation}
	where we used the definition of $L$, and we choose that Assumption~\ref{ass:smooth} holds with probability at least $1 - \frac{\delta}{2}$.
	This concludes the proof.
\end{proof}

%

\section{Full Description of the \emph{Hit-the-Spitfire} Security Game}\label{sec_ppendix_missile}

In this section, we provide a complete description of the \emph{Hit-the-Spitfire} security game introduced in Section~\ref{sec:experiments} (see also Figure~\ref{fig:gameinstance}).
In such game, the first player is a terrestrial counter-air defensive unit that has to hit an approaching enemy airplane with a surface-to-air heat-seeking missile.

The defender becomes aware of the airplane presence when it is $h_\perp$ meters far away.
Then, immediately after detection time, the counter-air unit has to launch a missile with an angle $\theta \in [0,1]$ (in radians) with respect to the direction along which the airplane has been detected.
The maximum angle is determined by the finite amount of fuel in the missile.

We assume, for simplicity, that the airplane is moving along the direction perpendicular to the line passing through the defensive unit and the point where the plane is detected, with a constant speed $v_d$.
The missile moves along the straight line determined by $\theta$ with a constant speed $v_a > v_d$.

The airplane can deploy a flare with the intent of deflecting the missile, and it has to choose where to do it, expressed in terms of distance $s \in [0,s_{\max}]$ measured starting from the point where the airplane has been detected.
Notice that $s_{\max}$ represents the maximum distance from the detection point at which the missile can hit the plane, given the amount of fuel available.
Formally,
$$s_{\max} = \frac{v_d h}{v_a \cos(1)}.$$
The flare acts at a distance of $h_f$ meters from the airplane.
We denote with $\ell$ the length of the airplane (in meters), while the flare covers an area large $\frac{\ell}{2}$ meters around the point where it is deployed.

When the missile hits the airplane, the latter suffers a damage $d \in [0,1]$, which represents the defender's utility $u$ (or, equivalently, the opposite of the airplane utility).
The damage is maximal if the airplane is hit at its middle point, while it decreases if the hitting point moves towards the extremes.
Formally:
\begin{equation*}
	d(x) = - \frac{4}{\ell^2} x^2 + 1,
\end{equation*}
where $x$ is the hitting point (in meters) starting from the tail of the plane.

For simplicity, we equivalently express the second player's strategy $s$ as $s_d = \frac{s}{s_{\max}}$, so that $s_d \in [0,1]$.

The game outcome is determined as follows.
If the missile trajectory does not intercept the flare and intercepts the airplane shape, then the missile hits the plane, dealing damage..
If the missile intercepts the flare, then it is deflected away from the airplane with some probability.
Specifically, it still hits the airplane with a probability proportional to the angle between the flare and the missile.
For the sake of simplicity, we model this process as a Bernoulli with mean $1 - \frac{|s_d s_{\max} - x_d|}{s_{\max}}$, where $x_d$ is the position of the airplane (measured starting from the detection point) when the flare was hit.
In the case of a hit, the hitting point is uniformly chosen over the plane as $U \left( [-\frac{\ell}{2},\frac{\ell}{2}] \right)$.

The overall process of computing the utility $u$ is depicted in Algorithm~\ref{alg:spit}.

\paragraph{Intercepting the Flare}
The missile intercepts the flare after the time necessary to travel the vertical distance $y_a = h_\perp - h_f$.
%
Since its vertical speed is $v_a \cos(\theta)$, such time is $t = \frac{y_a}{v_a \cos(\theta)} = \frac{h_\perp - h_f}{v_a \cos(\theta)}$.
The space covered in the horizontal direction by the missile is now $x_a = (h - h_f) \tan{\theta}$.
Therefore, we have that the missile hits the flare when:
\begin{align}
&s_d s_{\max} - \frac{\ell}{4} \leq x_a \leq s_d s_{\max} + \frac{\ell}{4},\\
&s_d s_{\max} - \frac{\ell}{4} \leq (h_\perp - h_f) \tan{\theta} \leq s_d s_{\max} + \frac{\ell}{4}.
\end{align}

\paragraph{Intercepting the Airplane}
Assuming the flare did not intercept the missile, the latter hits the airplane after the time necessary to travel the vertical distance $y_a = h$.
Since its vertical speed is $v_a \cos(\theta)$, such time is $t = \frac{y_a}{v_a \cos(\theta)} = \frac{h}{v_a \cos(\theta)}$.
Thus, te missile intercepts the plane if:
\begin{align}
& x_d - \frac{\ell}{2} \leq x_a \leq x_d + \frac{\ell}{2},\\
& v_d \frac{h_\perp}{v_a \cos(\theta)} - \frac{\ell}{2} \leq h_\perp \tan(\theta) \leq v_d \frac{h_\perp}{v_a \cos(\theta)} + \frac{\ell}{2}.
\end{align}

\begin{algorithm}[t!]
	\caption{Interaction of the \emph{Hit-the-Spitfire} Game}
	\label{alg:spit}
	\begin{algorithmic}
		\State The first player (defender) selects $\theta \in [0, 1]$
		\State The second player (airplane) selects $s_d \in [0, 1]$ (or, equivalently, $s \in [0,s_{\max}]$)
		\If {$s_d s_{\max} - \frac{\ell}{4} \leq (h_\perp - h_f) \tan{\theta} \leq s_d s_{\max} + \frac{\ell}{4}$}
		\State Generate a sample $p$ from $Be \left( 1 - \frac{|s_d s_{\max} - x_d|}{s_{\max}} \right)$
		\If{$p = 1$}
		\State Generate a sample $q$ from a uniform distribution $U([-\frac{\ell}{2}, \frac{\ell}{2}] )$
		\State Deal damage $d = d(q)$
		\Else
		\State Deal damage $d = 0$ 
		\EndIf
		\Else
		\If {$v_d \frac{h_\perp}{v_a \cos(\theta)} - \frac{\ell}{2} \leq h_\perp \tan(\theta) \leq v_d \frac{h_\perp}{v_a \cos(\theta)} + \frac{\ell}{2}$}
		\State Deal damage $d = d \left( \left| \frac{v_d}{v_a \cos(\theta)} - h_\perp \tan(\theta) \right| \right)$
		\Else
		\State Deal damage $d = 0$ 
		\EndIf
		\EndIf
	\end{algorithmic}
\end{algorithm}

%

\section{Additional Experimental Results} \label{sec:add_res}

\begin{figure}[h!]
	\subfloat[$\epsilon$ vs. $\hat{\epsilon}$ for different values of $K_\epsilon$ (Mat\'ern kernel with $\nu = 1.5$).]{
\begin{tikzpicture}

\definecolor{color0}{rgb}{0.75,0.75,0}

\begin{axis}[
legend cell align={left},
legend style={at={(0.80,0.97)}, anchor=north west, draw=white!80.0!black},
tick align=outside,
tick pos=both,
x grid style={lightgray!92.02614379084967!black},
xlabel={$K_\epsilon$},
xmin=0, xmax=55,
xtick style={color=black},
y grid style={lightgray!92.02614379084967!black},
ymin=0.09, ymax=5.8,
ytick style={color=black},
ymode=log
]

\addplot [semithick, green!50.0!black, mark=square]
table {%
3	4.402105774
5	1.467368591
8	0.733684296
14	0.366842148
26	0.183421074
};
\addlegendentry{$\epsilon$}
\addplot [semithick, blue, mark=o]
table {%
3	2.028144278
5	0.150893888
8	0.150893888
14	0.150407506
26	0.018965974
};
\addlegendentry{$\hat{\epsilon}$}

\end{axis}

\end{tikzpicture}}
	\subfloat[$\epsilon$ vs. $\hat{\epsilon}$ for different values of $K_\epsilon$ (squared exponential kernel with $l=1.5$).]{\begin{tikzpicture}

\definecolor{color0}{rgb}{0.75,0.75,0}

\begin{axis}[
legend cell align={left},
legend style={at={(0.80,0.97)}, anchor=north west, draw=white!80.0!black},
tick align=outside,
tick pos=both,
x grid style={lightgray!92.02614379084967!black},
xlabel={$K_\epsilon$},
xmin=0, xmax=55,
xtick style={color=black},
y grid style={lightgray!92.02614379084967!black},
ymin=0.00007, ymax=5.8,
ytick style={color=black},
ymode=log
]

\addplot [semithick, green!50.0!black, mark=square]
table {%
3	4.773565692
5	1.591188564
8	0.795594282
14	0.397797141
};
\addlegendentry{$\epsilon$}
\addplot [semithick, blue, mark=o]
table {%
3	1.38855313
5	0.046471654
8	0.0001
14	0.0001
};
\addlegendentry{$\hat{\epsilon}$}

\end{axis}

\end{tikzpicture}}
\end{figure}

\end{document}